\documentclass[journal]{IEEEtran}

\usepackage[english]{babel}
\usepackage[utf8x]{inputenc}
\usepackage[T1]{fontenc}

\usepackage{hyperref}
\hypersetup{hidelinks}       %
\usepackage{url}            %
\usepackage{booktabs}       %
\usepackage{amsfonts}       %
\usepackage{nicefrac}       %
\usepackage{microtype}      %
\usepackage{xcolor}         %
\usepackage{algorithm}
\usepackage{algorithmic}

\usepackage{amsmath}
\usepackage{amssymb}
\usepackage{mathtools}
\usepackage{amsthm}

\usepackage[font=footnotesize]{caption}

\usepackage[capitalize,noabbrev]{cleveref}

\theoremstyle{plain}
\newtheorem{theorem}{Theorem}

\newtheorem{lemma}{Lemma}

\theoremstyle{definition}
\newtheorem{definition}{Definition}
\newtheorem{assumption}{Assumption}
\theoremstyle{remark}

\usepackage{psfrag,epsfig,graphics}
\usepackage{amsmath,amsthm,amssymb,multirow}
\usepackage{mathbbol}
\usepackage{amssymb}
\usepackage{threeparttable,balance}
\usepackage{subcaption}
\usepackage{cuted}

\usepackage{empheq}
\newcommand{\boxedeq}[2]{\begin{empheq}[box={\fboxsep=6pt\fbox}]{align}\label{#1}#2\end{empheq}}

\DeclareSymbolFontAlphabet{\amsmathbb}{AMSb}%

\newcommand{\bA}{\boldsymbol{A}}
\newcommand{\bB}{\boldsymbol{B}}
\newcommand{\bC}{\boldsymbol{C}}

\newcommand{\bG}{\boldsymbol{G}}

\newcommand{\bI}{\boldsymbol{I}}

\newcommand{\bM}{\boldsymbol{M}}

\newcommand{\bO}{\boldsymbol{O}}
\newcommand{\bP}{\boldsymbol{P}}

\newcommand{\bS}{\boldsymbol{S}}

\newcommand{\bU}{\boldsymbol{U}}
\newcommand{\bV}{\boldsymbol{V}}

\newcommand{\bX}{\boldsymbol{X}}
\newcommand{\bY}{\boldsymbol{Y}}

\newcommand{\bone}{\boldsymbol{1}}

\newcommand{\bs}{\boldsymbol{s}}

\newcommand{\bx}{\boldsymbol{x}}
\newcommand{\bw}{\boldsymbol{w}}

\newcommand{\calB}{\mathcal{B}}

\newcommand{\calF}{\mathcal{F}}
\newcommand{\calG}{\mathcal{G}}

\newcommand{\calO}{\mathcal{O}}

\newcommand{\calS}{\mathcal{S}}

\newcommand{\calM}{\mathcal{M}}
\newcommand{\calN}{\mathcal{N}}

\newcommand{\bbE}{\amsmathbb{E}}

\newcommand{\bbR}{\amsmathbb{R}}

\newcommand{\GR}{\calG_n^p}
\newcommand{\ST}{\calS_n^p}
\newcommand{\OG}{\calO_p}
\newcommand{\bzero}{\boldsymbol{0}}

\newcommand{\bdelta}{\boldsymbol{\delta}}

\newcommand{\bphi}{\boldsymbol{\phi}}
\newcommand{\bvarphi}{\boldsymbol{\varphi}}

\newcommand{\btheta}{\boldsymbol{\theta}}

\newcommand{\bxi}{\boldsymbol{\xi}}

\newcommand{\bSigma}{\boldsymbol{\Sigma}}

\newcommand{\cb}[1]{\boldsymbol{#1}}
\newcommand{\tr}{\emph{tr}}

\newcommand{\diag}{\mathrm{diag}}
\newcommand{\col}{\mathrm{col}}
\newcommand{\inj}{\mathrm{inj}}

\begin{document}

\title{Distributed Riemannian Optimization in Geodesically Non-convex Environments}

\author{Xiuheng~Wang,
        Ricardo~Borsoi,
         C\'edric~Richard,
         and Ali~H.~Sayed%
}

\maketitle

\begin{abstract}
This paper studies the problem of distributed Riemannian optimization over a network of agents whose cost functions are geodesically smooth but possibly geodesically non-convex.
Extending a well-known distributed optimization strategy called diffusion adaptation to Riemannian manifolds, we show that the resulting algorithm, the \textit{Riemannian diffusion adaptation},
provably exhibits several desirable behaviors when minimizing a sum of geodesically smooth non-convex functions over manifolds of bounded curvature.
More specifically, we establish that the algorithm can approximately achieve network agreement in the sense that Fr\'echet variance of the iterates among the agents is small.
Moreover, the algorithm is guaranteed to converge to {a neighborhood of} a first-order stationary point for general geodesically non-convex cost functions.
When the global cost function additionally satisfies the {local} Riemannian Polyak-Lojasiewicz (PL) condition, we also show that it converges linearly under a constant step size up to a steady-state error.
Finally, we apply this algorithm to decentralized robust principal component analysis (PCA) {and low-rank matrix completion problems} and illustrate its convergence and performance through numerical simulations.
\end{abstract}


\section{Introduction}
In the decentralized setting, this work considers geodesically non-convex (g-non-convex) problems where $K$ agents cooperate to solve the following optimization problem over a Riemannian manifold $\calM$:
\begin{equation}
    \label{eq:optimization}
	\min_{w\in \calM} \frac{1}{K}\sum_{k=1}^K J_k(w)\,,
\end{equation}
where $J_k: \calM \to \bbR$ is a local cost function defined for each agent by $J_k(w) = \bbE_{\bx_k}\big\{Q(w;\bx_k)\big\}$ in terms of the expectation of some loss function $Q(w;\bx_k)$.
The expectation in $J_k(w)$ is computed over the unknown distribution of the data $\bx_k$, which makes it necessary to use a stochastic approximation based on a set of independent realizations $\bx_{k,t}$, observed sequentially over time.
A wide range of applications in machine learning, signal processing, and control can be written in the form of~\eqref{eq:optimization}. For instance, principal component analysis (PCA) can be formulated as minimizing the negative projected variance over the Grassmann manifold~\cite{cunningham2015linear}. Gaussian mixture model inference involves optimizing the log-likelihood function over the manifold of symmetric positive definite matrices~\cite{hosseini2015matrix, collas2023riemannian}. Similarly, low-rank matrix completion seeks to minimize the reconstruction error on the manifold of fixed-rank matrices~\cite{boumal2011rtrmc, vandereycken2013low}. In each of these decentralized settings, the local cost function $J_k$ is defined based on the data available to agent $k$.

\subsection{Related work}
\begin{table*}[t]
\caption{Comparison of modeling assumptions and results for stochastic gradient-based methods. Statements marked with $*$ are for extrinsic methods, based on specific embeddings of the manifolds in Euclidean space. The works marked with $\dag$ establish dynamic regret results.}
\centering
\begin{tabular}{cccccccc}
		\hline
		{} & Manifold & Convexity & Step size & Results \\ \hline
        \textbf{Centralized} \\
        \cite{bonnabel2013stochasticGradRiemannian} & Riemannian & g-convex & diminishing & asymptotic\\
        \cite{zhang2016firstOrderGeodesicallyConvex} & Hadamard & {g-convex/g-strongly-convex} & diminishing & non-asymptotic \\
        \cite{tripuraneni2018averaging} & Riemannian & g-strongly-convex & diminishing & asymptotic\\
        \cite{wang2021no} & Hadamard & g-convex & diminishing & non-asymptotic$^\dag$ \\
        \cite{hsieh2023riemannian} & Riemannian & g-non-convex & diminishing & asymptotic\\
        \cite{wang2024nonparametric} & Hadamard & g-strongly-convex & constant & non-asymptotic\\

        \textbf{Decentralized} \\
        \cite{wang2023incremental} & {unit sphere} & non-convex$^*$ & constant & non-asymptotic \\
        \cite{chen2021decentralized} & Stiefel & non-convex$^*$ & diminishing & non-asymptotic \\
        \cite{sun2024retraction} & {Stiefel} & {non-convex$^*$} & {constant} & {non-asymptotic} \\
        \cite{deng2025decentralized} & compact submanifolds & non-convex$^*$ & diminishing & non-asymptotic \\
        \cite{shah2017distributed} & {Riemannian} & {g-convex} & {diminishing} & {asymptotic} \\
        \cite{wang2025riemannian} & Riemannian & g-convex & constant & non-asymptotic \\
        \cite{chen2024decentralized_dynamic} & Hadamard & g-convex & constant & non-asymptotic$^\dag$ \\
        \cite{sahinoglu2025decentralized} & Riemannian & g-convex & diminishing & non-asymptotic$^\dag$ \\
        \textbf{This work} & \textbf{Riemannian} & \textbf{g-non-convex} & \textbf{constant} & \textbf{non-asymptotic}\\
		\hline
	\end{tabular}
\label{tab:comparison}
\end{table*}

Motivated by these applications, recent works have pursued the study of stochastic optimization algorithms on Riemannian manifolds, both in the centralized and decentralized settings~\cite{bonnabel2013stochasticGradRiemannian}-\cite{sahinoglu2025decentralized}.
The first asymptotic convergence analysis for Riemannian stochastic gradient descent (R-SGD) was presented in~\cite{bonnabel2013stochasticGradRiemannian} under diminishing step sizes. This was followed by the first non-asymptotic convergence guarantees for first-order Riemannian optimization in both geodesically convex (g-convex) and g-non-convex environments in~\cite{zhang2016firstOrderGeodesicallyConvex}. A variant of R-SGD~\cite{tripuraneni2018averaging} generalizing the classical Polyak-Ruppert iterate-averaging scheme was then developed and analyzed for geodesically strongly convex (g-strongly-convex) cases. The Riemannian online optimization was considered in~\cite{wang2021no}, studying the dynamic regret for g-convex functions on Hadamard manifolds. More recently,~\cite{hsieh2023riemannian} investigated the behavior of stochastic algorithms around saddle points in g-non-convex settings. Finally,~\cite{wang2024nonparametric} established the non-asymptotic convergence of R-SGD with a constant step size in the g-strongly-convex setting, applying the findings to change point detection on manifolds.

The literature on decentralized Riemannian optimization is broadly categorized into \textit{extrinsic} and \textit{intrinsic} methods.
Extrinsic methods, which are based on the \textit{induced arithmetic mean}~\cite{sarlette2009consensus}, rely on embedding the manifold in a Euclidean space. This dependency often restricts their application to specific manifolds. {Early work considered network agreement problems on Stiefel manifolds~\cite{chen2021local}. For stochastic optimization, various strategies have been adapted to decentralized R-SGD-type algorithms on the unit sphere~\cite{wang2023incremental}, Stiefel manifolds~\cite{chen2021decentralized, sun2024retraction}, and compact submanifolds~\cite{deng2025decentralized}, with non-convex settings in the Euclidean space.}
In contrast, intrinsic methods are developed using the manifold's inherent geometry, leveraging tools like the \textit{Fréchet mean}~\cite{tron2012riemannian} (see Definition~\ref{definition:frechet_mean_variance}), geodesic distance, and exponential mapping. This approach allows them to be applied to a more general class of manifolds.
Several distributed strategies have been developed within this framework. Early work focused on achieving network agreement~\cite{tron2012riemannian,kraisler2023distributed}, while another approach solved g-convex optimization problems using a diminishing step size~\cite{shah2017distributed}.

A more recent line of work has focused on the diffusion adaptation strategy~\cite{chen2012diffusion, sayed2013diffusion}.
This strategy was first extended to general Riemannian manifolds in~\cite{wang2024riemannian}, and a more efficient algorithm was subsequently proposed in~\cite{wang2025riemannian} with convergence guarantees for g-convex costs.
Concurrent work established a dynamic regret bound for this algorithm on Hadamard manifolds~\cite{chen2024decentralized_dynamic}, which was very recently extended beyond the Hadamard setting in~\cite{sahinoglu2025decentralized}. However, both results only work for g-convex costs.
For ease of reference, we summarize the modeling conditions and results from related works in Table~\ref{tab:comparison}.

\subsection{Contributions}
To the best of our knowledge, this work marks the first study into the analysis of distributed Riemannian optimization methods for g-non-convex costs {on general manifolds}. The key contributions of this work are threefold.
First, we establish that the Riemannian diffusion adaptation algorithm approximately reaches consensus over the network in sufficient iterations (Theorem~\ref{theorem:network_agreement}).
Second, we show that the iterates of networked agents converge to {a neighborhood of} a first-order stationary point in general g-non-convex cases (Theorem~\ref{theorem:convergence}). Furthermore, under an additional {local} Riemannian Polyak-Lojasiewicz (PL) condition, we show the iterates converge linearly to a {local} optimum up to a steady-state error (Theorem~\ref{theorem:convergence_pl}).
Finally, we formulate decentralized robust PCA {and low-rank matrix completion problems on manifolds as examples} of g-non-convex optimization and illustrate the convergence and performance of the algorithm.

The rest of this paper is organized as follows. Section~\ref{sec:background} introduces Riemannian geometry and optimization tools. Section~\ref{sec:algorithm_assumptions} presents the algorithm and the modeling conditions. Sections~\ref {sec:agreement} and~\ref {sec:convergence} present the main theoretical results on network agreement and non-asymptotic convergence, respectively, followed by {applications in Section~\ref{sec:applications}} and numerical experiments in Section~\ref{sec:simulations}. Finally, Section~\ref{sec:conclusion} concludes the paper.

\section{Background}
\label{sec:background}
This section briefly introduces some basic concepts of Riemannian geometry~\cite{lee2006riemannian, do2016differential}, focusing on the essential tools for manifold optimization~\cite{absil2009manoptBook, boumal2023introduction}.

A \textit{Riemannian manifold} $(\calM, g)$ is a constrained set $\calM$ endowed with a \textit{Riemannian metric} $g_x(\cdot, \cdot):T_x\calM\times T_x\calM\to\bbR$, defined for every point $x\in\calM$, with $T_x\calM$ denoting the so-called \textit{tangent space} of $\calM$ at $x$.
A \textit{geodesic} $\gamma_v:[0,1]\to\calM$ is the curve of minimal length linking two points $x,y\in\calM$ such that $x=\gamma(0)$ and $y=\gamma(1)$, with $v\in T_x\calM$ the velocity of $\gamma_v$ at $0$ denoted by $\dot{\gamma}_v(0)$. The \textit{geodesic distance} $d_{\calM}(\cdot\,,\cdot):\calM\times\calM\to\bbR$ is defined as the length of the geodesic linking two points $x,y\in\calM$.

The \textit{exponential map} $w=\exp_{x}(v)$ is defined as the point $w\in\calM$ located on the unique geodesic $\gamma_v(t)$ with endpoints $x=\gamma_v(0)$, $w=\gamma_v(1)$ and velocity $v=\dot{\gamma}_v(0)$.
Consider a smooth function $f:\calM \to \bbR$. The \textit{Riemannian gradient} of $f$ at $x\in\calM$ is defined as the unique tangent vector $\nabla f(x)\in T_x\calM$ satisfying
	$\frac{d}{dt}\big|_{t=0}f(\exp_{x}(tv)) = \langle\nabla f(x), v\rangle_x$,
for all $v\in T_x\calM$.
For a smooth map $F:\calM\to\calN$ between two manifolds, the
\textit{differential} of $F$ at $x\in\calM$ is the linear map
$DF(x):T_x\calM \to T_{F(x)}\calN$, defined as $
DF(x)[v] = \tfrac{d}{dt}\big|_{t=0} F(\exp_x(tv))$
for any $v\in T_x\calM$.
The \textit{second differential} of $F$ at $x$ is the bilinear map $D^2F(x):T_x\calM\times T_x\calM \to T_{F(x)}\calN$ defined as
$D^2F(x)[u,v] = \tfrac{d}{dt}\big|_{t=0} DF(\exp_x(tu))[v]$ for any $u,v\in T_x\calM$.
We further define \textit{parallel transport} $\Gamma_x^y: T_x\calM \to T_y\calM$ as the map transporting a vector $v\in T_x\calM$ to $T_y\calM$ along a path $\exp_x(v)$ connecting $x$ to $y = \exp_x(v)$ such that the induced vector fields are parallel.
The map $\Gamma_x^y$ is an isometry.
We also consider a \textit{vector transport} map $\Lambda_x^y: T_x\calM \to T_y\calM$ defined as the differential $D\exp_x(\exp_x^{-1}(y))$ of the exponential mapping. Note that $\exp_x(\cdot)$ can be regarded as one particular \textit{retraction}.
{The definition of \textit{sectional curvature} $\kappa$ is the Gauss curvature of a two-dimensional submanifold formed as the
image of a two-dimensional subspace of a tangent space after exponential mapping~\cite{zhang2016firstOrderGeodesicallyConvex}.}

\section{The algorithm and assumptions}
\label{sec:algorithm_assumptions}

\subsection{Riemannian diffusion adaptation}
\label{ssec:algorithm}

For the case of the optimization problem~\eqref{eq:optimization} when $\calM$ is the Euclidean space, the well-known diffusion adaptation strategy has been proposed in~\cite{chen2012diffusion, sayed2013diffusion} and demonstrated in~\cite{sayed2014adaptationLearningOptimization,chen2015learning} to offer improved performance and stability guarantees under constant step size learning and adaptive scenarios.
Recently, this strategy has been extended to Riemannian manifolds in~\cite{wang2024riemannian}, and an efficient algorithm with convergence guarantees was proposed in~\cite{wang2025riemannian} as follows:
\boxedeq{eq:diffusion}{
\begin{split}
    \bphi_{k, t} &= \exp_{\bw_{k, t-1}}\big( - \mu \widehat{\nabla J}_k(\bw_{k, t-1})\big) \,, \\
    \bw_{k, t} & = \exp_{\bphi_{k, t}}\left(\alpha \sum_{\ell=1}^K c_{\ell k} \exp^{-1}_{\bphi_{k, t}}(\bphi_{\ell, t})\right)\,,
\end{split}}
where $\mu$ and $\alpha$ are constant step sizes, $\widehat{\nabla J}_k$ is the Riemannian stochastic gradient, where the expectation is approximated by the independent realization $\bx_{k,t}$. The Riemannian diffusion adaptation algorithm in~\eqref{eq:diffusion} contains two steps: an \textit{adaptation step} where agent $k$ uses R-SGD to update its solution $\bphi_{k, t}$ and a \textit{combination step} where the intermediate estimates $\{\bphi_{\ell, t}\}$ are combined, on the tangent space of $\bphi_{k, t}$, according to the weighting coefficients $\{c_{\ell k}\}$ to obtain the estimate $\bw_{k, t}$.
However, the behavior of the algorithm~\eqref{eq:diffusion} has only been theoretically studied under g-convex costs~\cite{wang2025riemannian, chen2024decentralized_dynamic, sahinoglu2025decentralized}. The main contribution of this work is the analysis of this algorithm in g-non-convex environments.

{The key difficulty in the g-non-convex case is that the intrinsic combination step is nonlinear, so neither linear consensus contractions in extrinsic schemes~\cite{wang2023incremental, chen2021decentralized, sun2024retraction, deng2025decentralized} nor g-convexity-based decoupling in~\cite{shah2017distributed, wang2025riemannian} applies. We therefore derive curvature-dependent bounds that track the Fr\'echet variance through both adaptation and combination steps in Lemma~\ref{lemma:frechet_variance_descent}, and control the consensus-bias term to close the non-asymptotic recursions via a Taylor expansion of composite exponential maps as in Lemma~\ref{lemma:taylor_expansion} and~\ref{lemma:consensus_bias_bound}.}

{To solve~\eqref{eq:optimization}, constrained diffusion strategies with projection or proximal steps have been developed to handle explicit constraints~\cite{nassif2017diffusion,alghunaim2019proximal}. In contrast, Riemannian diffusion adaptation operates directly on the manifold, where the update direction is computed on the tangent space, and the exponential map is used to map back to the manifold. As a result, the analysis of Riemannian diffusion adaptation must explicitly account for curvature and related geometric effects.}

\subsection{Modeling conditions}
\label{ssec:assumptions}

Following standard assumptions in the literature on distributed optimization~\cite{duchi2011dual, sayed2014adaptationLearningOptimization, yuan2016convergence, alghunaim2022unified}, we impose certain properties on the weighted adjacency matrix $C \triangleq [c_{\ell k}]$, which governs the interactions among agents represented as vertices of the graph $\calG$.
\begin{assumption}[\textbf{Regularization on graph}]
\label{assumption:graph}
Assume that the undirected graph $\calG$ is strongly connected and its adjacency matrix $C$ is symmetric and doubly stochastic, i.e. $c_{\ell k}\geq 0, \sum_{\ell=1}^K  c_{\ell k} = \sum_{k=1}^K  c_{\ell k} = 1$.
\end{assumption}
Under Assumption~\ref{assumption:graph}, we can recall the following lemma about the spectral properties of $C$ as in~\cite{duchi2011dual, sayed2014adaptationLearningOptimization}.
\begin{lemma}
\label{lemma:graph_mixing_rate}
Under Assumption~\ref{assumption:graph}, the adjacency matrix $C$ has a single eigenvalue at one, denoted by $\lambda_1=1$. Moreover, all other eigenvalues, denoted by $\{\lambda_k\}_{k=1}^K$, are strictly less than one in magnitude. The mixing rate $\lambda$ of the network is defined by:
\begin{align}
    \lambda \triangleq \rho\Big(C-\frac{1}{K}\bone\bone^T\Big) =\max_{k\in\{2, \cdots, K\}}|\lambda_k| < 1\,,
    \label{eq:graph_mixing_rate}
\end{align}
where $\rho(\cdot)$ denotes the spectral radius and $\bone$ represents the all-ones vector.
\end{lemma}
Let $\calB\subseteq\calM^K$ denote the {g-convex subset} of the product manifold $\calM^K$, {which is $K$-fold cartesian product of $\calM$ with itself.} We introduce the following standard assumptions {on $\calB$} in the literature on Riemannian optimization~\cite{afsari2011riemannianCenterOfMass, tron2012riemannian, bonnabel2013stochasticGradRiemannian, zhang2016fastRiemannianStochasticOptim, tripuraneni2018averaging, wang2025riemannian}.
\begin{assumption}[\textbf{Regularization on manifold}]
    \label{assumption:manifold}
   (a) The sequences $\{\bphi_{\ell,t}\}_{t\geq 0}$ and $\{\bw_{\ell,t}\}_{t\geq 0}$ generated by the algorithm stay continuously in $\calB$, and $J$ attains its optimum $\bw^*$ in $\calB$; (b) the sectional curvature in $\mathcal{B}$ is \emph{upper} bounded by $\kappa_{\max}$; (c) the sectional curvature in $\mathcal{B}$ is \emph{lower} bounded by $\kappa_{\min}$; and (d) $\mathcal{B}$ is compact, and the diameter of $\mathcal{B}$ is bounded by $B$, that is, $\max_{x,y\in\mathcal{B}} d(x,y) \le B$; (e) $B < B^*$, where $B^*$ is defined as $B^* \triangleq \min(\inj(\calM), \frac{\pi}{2\sqrt{\kappa_{max}}})$ with $\inj(\calM)$ the injectivity radius of $\calM$ {(see Definition 10.19 in~\cite{boumal2023introduction})}, which implies that the exponential map is invertible within $\mathcal{B}$.
\end{assumption}
Under Assumption~\ref{assumption:manifold}, we can recall the following lemma about trigonometric distance bounds, which is essential in the analysis of Riemannian optimization algorithms. This lemma is adapted from Proposition 2.1 of~\cite{jordan2022first}, which is composed of Lemma 5 of~\cite{zhang2016firstOrderGeodesicallyConvex} and Corollary 2.1 of~\cite{alimisis2020continuous}.
\begin{lemma}\label{lemma:trigonometric_distance}
Suppose that $a,b,c$ are the {sides (i.e., side lengths)} of a geodesic triangle in a Riemannian manifold with sectional curvature $\kappa$, and $A$ is the angle between sides $b$ and $c$ (defined through the inverse exponential map and inner product in tangent space). Then, we have:

(i) If  $\kappa$ is lower bounded by $\kappa_{\min}$, then
\begin{equation}
\label{eq:trigonometric_upper_bound}
a^2 \leq \zeta_1 \cdot b^2 + c^2 - 2bc\,\cos(A),
\end{equation}
where $\zeta_1 \triangleq B\sqrt{-\kappa_{\min}}\coth(B\sqrt{-\kappa_{\min}}) > 1$ if $\kappa_{\min} < 0$ and $\zeta_1 \triangleq 1$ if $\kappa_{\min} \geq 0$.

(ii) If  $\kappa$ is upper bounded by $\kappa_{\max} > 0$ and the diameter of $\calM$ is bounded by $\frac{\pi}{2\sqrt{\kappa_{\max}}}$, then
\begin{equation}
\label{eq:trigonometric_lower_bound}
a^2 \geq \zeta_2\cdot b^2 + c^2 - 2bc\cos(A),
\end{equation}
where $\zeta_2 \triangleq 1$ for $\kappa_{\max} \leq 0$ and $0 <\zeta_2 \triangleq B\sqrt{\kappa_{\max}}\cot(B\sqrt{\kappa_{\max}}) < 1$ for $\kappa_{\max} > 0$.
\end{lemma}
We also need the following lemma showing that both the exponential map and its inverse are Lipschitz.
\begin{lemma}\cite{sun2019escaping}
    \label{lemma:lipschitz_exponential_map}
    Let $x,y,z \in \calM$. Under Assumption~\ref{assumption:manifold}, we have:
\begin{equation}
    \frac{d(y,z)}{1+C_\kappa B^2}\leq {\|}\exp_x^{-1}(y)-\exp_x^{-1}(z) {\|}\leq(1+C_\kappa B^2)d(y,z)\,,
\end{equation}
where $C_\kappa$ is a constant depending on the curvature of the manifold $\calM$.
\end{lemma}
Meanwhile, we require the cost function $J_k$ at each agent to be geodesically smooth.
\begin{assumption}[\textbf{Geodesic smoothness}]
\label{assumption:smooth}
Assume the function $J_k$ is differentiable and geodesically $L$-smooth (i.e., its gradient is $L$-Lipschitz), that is, for any $x, y \in \calM$, it satisfies:
\begin{equation}
\label{eq:smooth}
    J_k(y) \leq J_k(x) + \langle \nabla J_k(x), \exp_x^{-1}(y)\rangle  + \frac{L}{2} \|\exp_x^{-1}(y)\|^2 \,,
\end{equation}
where the gradient of a function $J_k$ is said to be $L$-Lipschitz if, for any $x, y \in \calM$ in the domain of $J_k$, it satisfies:
\begin{equation}
\label{eq:lipschitz}
\big\|\nabla J_k(x) - \Gamma_y^x \nabla J_k(y) \big\| \leq L \, \|\exp_x^{-1}(y)\| \,.
\end{equation}
\end{assumption}
Under Assumptions~\ref{assumption:manifold} and~\ref{assumption:smooth}, we can recall the following lemma about the boundedness of the gradient.
\begin{lemma}\cite{wang2025riemannian}
\label{lemma:gradient_bound}
Under assumptions~\ref{assumption:manifold} and~\ref{assumption:smooth}, we have:
\begin{align}
    \label{eq:gradient_bound}
    \forall k\in\{1, \cdots, K\}\,,\quad
    \|\nabla J_k(\bw_{k, t})\| &\leq G\,,
\end{align}
for a non-negative constant $G < \infty$.
\end{lemma}
In addition, we make assumptions about the average and second moment of the gradient noise process.
\begin{assumption}[\textbf{Gradient noise process}]
\label{assumption:gradient_noise}
Denote $\calF_t$ as the filtration generated by the random process $\bw_{k, {i}}$ for all $k$ and for $i\leq t$, that is,
\begin{equation}
\label{eq:filtration}
    \calF_t \triangleq \{\bw_0, \bw_1, \cdots, \bw_{t}\} \,,
\end{equation}
where $\bw_{i} \triangleq \col\{\bw_{1,i}, \cdots, \bw_{K, i}\}$ contains the iterates across the network at time $i$.
For each agent $k$, define $\bs_{k, t+1}(\bw_{k, t}) \triangleq \widehat{\nabla J}_k(\bw_{k, t}) - {\nabla J_k}(\bw_{k, t})$ as the gradient noise process at the time instant $t$. It is assumed that
\begingroup
\begin{align}
\label{eq:gradient_mean}
    \bbE\{\bs_{k, t+1}(\bw_{k, t})|\calF_{t}\} &= \cb{0}\,,\\
    \label{eq:gradient_variance_square}
    \bbE\{\|\bs_{k, t+1}(\bw_{k, t})\|^2|\calF_{t}\} &\leq \sigma^2\,,
\end{align}
\endgroup
for some non-negative constant $\sigma$.
\end{assumption}
The above assumption implies that the gradient noise process is unbiased and has a bounded {second} moment, which is a common assumption in the stochastic optimization literature.

The previous assumptions will be used to analyze the algorithm in~\eqref{eq:diffusion} for general g-non-convex costs. Afterward, we will also study the convergence under the following {local Riemannian PL condition~\cite{han2025efficient, sun2024retraction}.}

\begin{assumption}[\textbf{Riemannian PL condition}]
    \label{assumption:pl_condition}
    Assume the global function $J$ is differentiable and satisfies {the local Riemannian PL condition for a submanifold $\calS \subset \calB$}, i.e., there exists a constant $\tau > 0$ such that for every {$\bw \in \calS$}, it holds that
    \begin{equation}
    \label{eq:pl_condition}
        J(\bw) - J(\bw^*) \leq \tau \|\nabla J(\bw)\|^2\,,
    \end{equation}
    where {$\bw^*$} is a {local minimum of $J$ on $\calS$}.
    \end{assumption}
The above {local Riemannian} PL condition is also called the $\tau$-gradient dominated condition, which implies that every stationary point is a {local} minimizer.
{Many problems formulated on manifolds are geodesically non-convex, including PCA, low-rank matrix completion, and Gaussian mixture models~\cite{zhou2019fasterSPIDER_GradientAveraging}. In particular, PCA problems are known to satisfy a (local) PL condition~\cite{liu2019quadratic, ye2021deepca, sun2024retraction}.}

\section{Network agreement}
\label{sec:agreement}
In this section, we show that the Riemannian diffusion adaptation algorithm in~\eqref{eq:diffusion} approximately converges toward network agreement after sufficient iterations.
For analysis purposes and ease of presentation, it is useful to introduce the following stacked vector notation, which collects variables from across the network as follows:
\begin{align}
     \bw_t & \triangleq \col\{\bw_{1,t}, \cdots,  \bw_{K,t}\}\,,
     \nonumber \\
      {\nabla J}(\bw_t) & \triangleq \col\left\{{\frac{1}{K}\nabla J_1}(\bw_{1,t}),  \cdots,  \frac{1}{K}{\nabla J}_K(\bw_{K,t})\right\}\,, \nonumber
\end{align}
where $\col\{\cdot\}$ denotes the column-wise stacking of its arguments. Note that $\bw_t \in \calM^K$ and ${\nabla J}(\bw_t) \in T_{\bw_t}\calM^K$ where $T_{\bw_t}\calM^K$ is the tangent space of $\calM^K$ at $\bw_t$ (see Proposition 3.20 in~\cite{boumal2023introduction}).
We also define the Fréchet mean and Fréchet variance~\cite{tron2012riemannian,afsari2011riemannianCenterOfMass,afsari2013convergence} as follows.
\begin{definition}[\textbf{Fréchet mean and variance}]
\label{definition:frechet_mean_variance}
Given a set of points $\{\bw_k\}_{k=1}^K$ on a Riemannian manifold $\calM$, the Fréchet mean $\bw_m$ is defined as the point that minimizes the sum of squared geodesic distances to all points, i.e.,
\begin{equation}
    \label{eq:frechet_mean}
    \bw_m \triangleq \arg\min_{\bw\in\calM} \sum_{k=1}^K d^2(\bw_k, \bw)\,.
\end{equation}
The Fréchet variance $V_F(\bw)$ is defined as the minimum value of the sum of squared geodesic distances to $\bw_m$, i.e.,
\begin{equation}
    \label{eq:frechet_variance}
    V_F(\bw) \triangleq \sum_{k=1}^K d^2(\bw_k, \bw_m)\,.
\end{equation}
\end{definition}
{Under Assumption~\ref{assumption:manifold}, the Fr\'echet mean in~(15) exists and is unique~\cite{afsari2011riemannianCenterOfMass, tron2012riemannian}.}
In addition, the combination step in~\eqref{eq:diffusion} can be regarded as one-step Riemannian gradient descent on the consensus bias~\cite{wang2025riemannian, tron2012riemannian}, as defined below.

\begin{definition}[\textbf{Consensus bias}]
\label{definition:consensus_bias}
Given a set of points $\{\bphi_k\}_{k=1}^K$ over the network with a weighted adjacency matrix $C$, the consensus bias $P(\bphi)$ is defined as the sum of weighted squared geodesic distances between all pairs of points, i.e.,
\begin{equation}
    \label{eq:consensus_bias}
    P(\bphi) \triangleq \sum_{k=1}^K\sum_{\ell=1}^K c_{\ell k} d^2(\bphi_k, \bphi_\ell)\,.
\end{equation}
\end{definition}

\subsection{Fréchet variance reduces on the combination step}
Based on these definitions, we first study the behavior of the Fréchet variance of the solutions on the combination step in~\eqref{eq:diffusion}. We start by establishing a lemma, which relates the Fréchet variance of the variables $\{\bw_{k, t}\}_{k=1}^K$ to that of $\{\bphi_{k, t}\}_{k=1}^K$ in~\eqref{eq:diffusion}. The following lemma builds on Assumption~\ref{assumption:graph} and Lemma~\ref{lemma:trigonometric_distance} under the additional condition set forth in Assumption~\ref{assumption:manifold}.
\begin{lemma}
\label{lemma:frechet_variance_descent_one_iteration}
Under Assumption~\ref{assumption:graph} and~\ref{assumption:manifold}, suppose $\alpha \in (0, \frac{\zeta_2}{\zeta_1})$. The Fréchet variances of $\{\bphi_{k, t}\}_{k=1}^K$ and $\{\bw_{k, t}\}_{k=1}^K$ satisfy
\begin{align}
\label{eq:frechet_variance_descent_one_iteration}
    V_F(\bw_{t})  &\leq V_F(\bphi_{t}) + (\zeta_1\alpha^2 - \zeta_2\alpha) P(\bphi_{t}) \,,
\end{align}
where the constant $\zeta_1\alpha^2 - \zeta_2\alpha < 0$.
\end{lemma}
\begin{proof}
See Appendix~\ref{appx:lemma:frechet_variance_descent_one_iteration}.
\end{proof}
From this lemma, we can see that the Fréchet variance of $\{\bw_{k, t}\}_{k=1}^K$ is reduced in comparison to that of $\{\bphi_{k, t}\}_{k=1}^K$ after the combination step in~\eqref{eq:diffusion} since the constant $\zeta_1\alpha^2 - \zeta_2\alpha$ is negative. This indicates that the combination step helps to reduce the network disagreement among agents.
To further study the reduction of the Fréchet variance, we introduce the following lemma, which is obtained by lower bounding the consensus bias term $P(\bphi_{t})$ in Lemma~\ref{lemma:frechet_variance_descent_one_iteration} as a function of $V_F(\bphi_{t})$ using Lemma~\ref{lemma:lipschitz_exponential_map} and Lemma~\ref{lemma:graph_mixing_rate}.
\begin{lemma}
\label{lemma:graph_topology}
Under Assumptions~\ref{assumption:graph} and~\ref{assumption:manifold}, suppose $\alpha \in (0, \frac{\zeta_2}{\zeta_1})$. The Fréchet variances of $\{\bphi_{k, t}\}_{k=1}^K$ and $\{\bw_{k, t}\}_{k=1}^K$ satisfies the relation
\begin{align}
    \label{eq:graph_topology}
        V_F(\bw_{t})
    \leq \left(1 - \varepsilon \right) V_F(\bphi_{t})
    \,,
\end{align}
where
\begin{align}
    \label{eq:varepsilon}
    \varepsilon \triangleq -\frac{2(1-\lambda)\left(\zeta_1\alpha^2 - \zeta_2\alpha\right)}{(1+C_\kappa B^2)^2} {\in (0, \frac{1}{2}]}\,,
\end{align}
denotes a constant term with $\lambda$ defined in \eqref{eq:graph_mixing_rate}.
\end{lemma}
\begin{proof}
See Appendix~\ref{appx:lemma:graph_topology}.
\end{proof}
From Lemma~\ref{lemma:graph_topology}, we observe that the Fréchet variance decreases by a multiplicative factor of $(1 - \varepsilon)$ on the combination step, where $\varepsilon > 0$ depends on the graph topology, the manifold curvature, and the step size $\alpha$.
For example, the reduction in Fréchet variance is more significant when the network is densely connected, as the spectral gap $(1 - \lambda)$ is large. If the sectional curvature is zero (the manifold is flat), we may set $C_{\kappa}=0$ and $\zeta_1=\zeta_2=1$, in which case $\varepsilon =2(1-\lambda)\alpha(1-\alpha)$.
This lemma shows that the combination step in~\eqref{eq:diffusion} contributes to a linear Fréchet variance reduction among agents. The proof of this lemma is partially inspired by the results in~\cite{chen2024decentralized_dynamic,sahinoglu2025decentralized}.

\subsection{Evolution of Fréchet variance over iterations}
We next use Lemma~\ref{lemma:graph_topology} to show the evolution of the Fréchet variance of $\{\bw_{k, t}\}_{k=1}^K$ over different iterations $t\ge0$. In the following lemma, we relate the Fréchet variance of $\{\bw_{k, t}\}_{k=1}^K$ to that of the previous iteration $\{\bw_{k, t-1}\}_{k=1}^K$ based on Lemma~\ref{lemma:graph_topology} and the adaptation step in~\eqref{eq:diffusion}. The proof of this lemma builds on the assumptions in Lemma~\ref{lemma:graph_topology} and Lemma~\ref{lemma:gradient_bound} under the additional conditions on the gradient and its noise in Assumption~\ref{assumption:gradient_noise}.
\begin{lemma}
\label{lemma:frechet_variance_descent}
Under assumptions~\ref{assumption:graph}--\ref{assumption:gradient_noise}, suppose $\alpha \in (0, \frac{\zeta_2}{\zeta_1})$. The sequence of Fréchet variances $\{V_F(\bw_{t})\}_{t\geq 0}$ satisfies the relation
\begin{align}
    \label{eq:frechet_variance_descent}
    \bbE V_F(\bw_{t})
    \leq & \left(1 - {\frac{\varepsilon}{2}}\right) \bbE V_F(\bw_{t-1})  \nonumber \\ & \, + \left(1 - \varepsilon\right) \mu^2K\big(2\zeta_1 G^2 + {2(1-\varepsilon)}\varepsilon^{-1} G^2 \nonumber \\ & \, \qquad + 2\zeta_1 \sigma^2\big)\,,
\end{align}
with $\varepsilon$ defined in~\eqref{eq:varepsilon}.
\end{lemma}
\begin{proof}
    See Appendix~\ref{appx:lemma:frechet_variance_descent}.
\end{proof}
This lemma reveals the evolution of the Fréchet variance over iterations.
The first term on the RHS of~\eqref{eq:frechet_variance_descent} is strictly smaller than $\bbE V_F(\bw_{t-1})$ by a factor {$(1 - \varepsilon/2) < 1$}, which suggests a decrease in the sequence of Fréchet variances. However, the second term on the RHS of~\eqref{eq:frechet_variance_descent} could potentially be large enough to allow this sequence to increase.
To address this, we demonstrate that with a small step size $\mu$ the Fréchet variance not only decreases strictly over iterations but also remains bounded above by a small value after enough iterations. This result is key to establishing the non-asymptotic agreement among the iterates.

\subsection{Agreement after sufficient iterations}
Building on the above analysis, we now present the main result of this section, which shows that the iterates $\{\bw_{k, t}\}_{k=1}^K$ achieve approximate network agreement after a sufficient number of iterations.
\begin{theorem}
\label{theorem:network_agreement}
Under assumptions~\ref{assumption:graph}--\ref{assumption:gradient_noise}, suppose $\alpha \in (0, \frac{\zeta_2}{\zeta_1})$. The sequence of Fréchet variances $\{V_F(\bw_{t})\}_{t\geq 0}$ satisfies the relation
\begin{align}
    \label{eq:network_agreement}
    \bbE V_F(\bw_{t}) & {}\leq{} {4}\left(1 - \varepsilon\right) \varepsilon^{-{1}} \mu^2{K}\big(2\zeta_1 G^2 + {2(1-\varepsilon)}\varepsilon^{-1} G^2 \nonumber \\ & \qquad {}+ 2\zeta_1 \sigma^2\big) \nonumber \\ & {}={} {\calO\big(\mu^2\alpha^{-2}(\zeta_2-\zeta_1\alpha)^{-2}\big)}\,,
\end{align}
with $\varepsilon$ defined in~\eqref{eq:varepsilon} and ${\calO\big(\mu^2\alpha^{-2}(\zeta_2-\zeta_1\alpha)^{-2}\big)}$ being a term that is equal to or higher in order than $\mu^{2}${, whose hidden constant is independent of $\mu$ and $\alpha$}, after a sufficient number of iterations $t_o$, which is given by
\begin{align}
    \label{eq:sufficient_iteration}
    t_o =  \frac{2\log(\mu)}{\log(1 - {\varepsilon/2})} + \calO(1) = {\calO\big(\mu^{-1}\alpha^{-1}(\zeta_2-\zeta_1\alpha)^{-1}\big)}\,,
\end{align}
for some small step sizes $\mu$.
\end{theorem}
\begin{proof}
    See Appendix~\ref{appx:theorem:network_agreement}.
\end{proof}
Theorem~\ref{theorem:network_agreement} guarantees that the iterates $\{\bw_{k, t}\}_{k=1}^K$ achieve approximate consensus across the network. Specifically, it shows that the value of $\bbE V_F(\bw_{t})$ (a measure of network disagreement) remains bounded by a term of order ${\calO\big(\mu^2\alpha^{-2}(\zeta_2-\zeta_1\alpha)^{-2}\big)}$ after $t_o$ iterations.
For a small step size {$\mu = o(\alpha)$}, this implies that the agents’ estimates can be made arbitrarily close.
\begingroup
A smaller $\mu$ improves the guaranteed agreement accuracy but lengthens the transient in~\eqref{eq:sufficient_iteration}. Both the agreement level and the transient depend on the combination step size, manifold curvature, and network topology through $\varepsilon$ in~\eqref{eq:varepsilon}.
\endgroup

{In contrast to the g-convex setting studied in~\cite{wang2025riemannian}, the absence of g-convexity makes the property $J(\bw_m)\leq J(\bw)$ no longer hold and prevents direct analysis of network agreement through the consensus bias term $\bbE P(\bphi_t)$. Instead, Theorem~\ref{theorem:network_agreement} characterizes network agreement through the Fréchet variance $\bbE V_F(\bw_t)$.} A key advantage of this formulation is that it reveals the influence of the network topology, captured by the spectral gap $1 - \lambda$, on both the achievable agreement level and the convergence time.
{The algorithm in~\eqref{eq:diffusion} yields non-exact agreement (small but nonzero Fr\'echet variance), which is often sufficient and efficient under constant step-sizes and stochastic gradients. For objectives requiring exact consensus, exact-agreement methods (e.g., exact diffusion~\cite{yuan2018exact} and EXTRA~\cite{shi2015extra}) can be preferable. Extending such methods to general manifolds remains an open direction.}
{In the next section, Section~\ref{sec:convergence}, the agreement bound of Theorem~\ref{theorem:network_agreement} is used through Lemma~\ref{lemma:consensus_bias_bound} to control the consensus-bias term in the cost descent inequality, thereby establishing the optimization convergence guarantees in Theorems~\ref{theorem:convergence} and~\ref{theorem:convergence_pl}.}

\section{Convergence analyses}
\label{sec:convergence}
In this section, we establish the convergence of the Riemannian diffusion adaptation algorithm in~\eqref{eq:diffusion} after a sufficient number of iterations $t_o$, for both the general g-non-convex case and the case satisfying the additional {local} Riemannian PL condition. To this end, we use the upper bound on $\bbE V_F(\bw_{t})$ derived in Theorem~\ref{theorem:network_agreement}.

\subsection{Descent inequality of the cost function}
We first introduce the following lemma, which establishes a key descent inequality that characterizes the expected decrease of the value of the cost function $J(\bw_t)$ over iterations.
\begin{lemma}
\label{lemma:cost_relation}
Under assumptions~\ref{assumption:graph}--\ref{assumption:gradient_noise}, suppose $\mu \in (0, \frac{1}{L}]$. The sequence $\{J(\bw_{t})\}_{t\geq 0}$ satisfies the following relation:
\begin{align}
    \label{eq:cost_relation}
    \bbE J(\bw_{t+1}) & {}\leq{} \bbE J(\bw_{t}) - \frac{\mu K}{4} \bbE \| {\nabla J}(\bw_{t})\|^2 \nonumber \\
    & \quad + \frac{9\alpha^2}{2\mu K} \bbE P(\bphi_{t+1}) {{}+ \mu^2L\sigma^2} \,.
\end{align}
\end{lemma}
\begin{proof}
    See Appendix~\ref{appx:lemma:cost_relation}.
\end{proof}
Compared with the centralized Riemannian SGD, see~\cite{bonnabel2013stochasticGradRiemannian, zhang2016firstOrderGeodesicallyConvex, wang2024nonparametric} for example, Lemma~\ref{lemma:cost_relation} has an additional consensus bias term $\bbE P(\bphi_{t+1})$. We therefore need to control the consensus bias term to ensure convergence of the cost function.
{The last term $\mu^2L\sigma^2$ on the RHS of~\eqref{eq:cost_relation} is a positive error term induced by the variance of the stochastic gradient noise. It is analogous to the higher-order noise terms arising in constant step-size stochastic gradient analyses~\cite{sayed2014adaptationLearningOptimization, bonnabel2013stochasticGradRiemannian, wang2024nonparametric}.}

\subsection{Bounding the consensus bias term}
{The absence of the g-convexity make the consensus bias term $\bbE P(\bphi_t)$ in Lemma~\ref{lemma:cost_relation} become more challenging to control.
To tackle this issue, we first introduce a technical lemma that provides a Taylor expansion of the composite exponential map as follows.}
\begin{lemma}
\label{lemma:taylor_expansion}
{Under assumptions~\ref{assumption:manifold}.} We have the following expansion:
\begin{align}
    \label{eq:taylor_expansion}
    \exp^{-1}_{\bw_{k,t}}(\bphi_{\ell,t+1}) & = \exp^{-1}_{\bw_{k,t}}(\bw_{\ell,t}) - \mu [\Lambda_{\bw_{k,t}}^{\bw_{\ell,t}}]^{-1} \nabla J_\ell(\bw_{\ell,t}) \nonumber \\ & \quad + {R_{\ell, t} + \bdelta_{\ell,t}}\,,
\end{align}
where
\begin{align}
\label{eq:residual_bound}
\|R_{\ell, t}\| \leq \mu^2 C_F {\|{\nabla J}_\ell(\bw_{\ell,t})\|^2}\,,
\end{align}
and
\begin{align}
\label{eq:perturbation_bound}
{\|\bdelta_{\ell,t}\|}
&{{}\leq (1+C_\kappa B^2)^2 \mu \|\bs_{\ell,t+1}(\bw_{\ell,t})\|\,,}
\end{align}
{where $C_F$ and $C_\kappa$ are constants depending on the local smoothness of the composite exponential maps and the curvature of the manifold, respectively.}
\end{lemma}

\begin{proof}
    See Appendix~\ref{appx:lemma:taylor_expansion}.
\end{proof}
Lemma~\ref{lemma:taylor_expansion} shows how to linearize the process of R-SGD iterate $\bphi_{\ell, t+1}$ in~\eqref{eq:diffusion} in the tangent space of $\bw_{k, t}$, which is crucial for bounding the consensus bias term in Lemma~\ref{lemma:consensus_bias_bound} (see the inequality~\eqref{eq:upper_bound_disagreement_2} in the proof of Lemma~\ref{lemma:consensus_bias_bound} for the precise relation).
The proof of Lemma~\ref{lemma:taylor_expansion} proceeds analogously to that of Lemma 4 in~\cite{tripuraneni2018averaging}, relying on the chain rule for differential mappings on manifolds.
We next use Lemma~\ref{lemma:taylor_expansion} to bound the consensus term as follows.
\begin{lemma}
\label{lemma:consensus_bias_bound}
Under assumptions~\ref{assumption:graph}--\ref{assumption:gradient_noise}, suppose $\alpha \in (0, \frac{\zeta_2}{\zeta_1})$. The network disagreement among all the local estimates can be bounded as follows:
\begin{align}
    \label{eq:consensus_bias_bound}
    & \quad \ \bbE P(\bphi_{t+1}) \nonumber \\ & \leq  {16}\left(1+C_{\kappa}B^2\right)^2\left(1+\mu^2L^2 + \mu^2 C_R^2 B^2\right)\bbE V_F(\bw_{t}) \nonumber \\ & \quad  + \calO(\mu^2)  \,,
\end{align}
with $C_R$ a constant depending on the smoothness of the exponential mapping.
\end{lemma}
\begin{proof}
    See Appendix~\ref{appx:lemma:consensus_bias_bound}.
\end{proof}
This lemma indicates that the consensus bias term $\bbE P(\bphi_{t+1})$ can be controlled by the expected Fréchet variance $\bbE V_F(\bw_{t})$ and some higher-order terms depending on the step size $\mu$.

\subsection{Convergence in general g-non-convex cases}
Based on the lemmas~\ref{lemma:gradient_bound},~\ref{lemma:cost_relation}, and~\ref{lemma:consensus_bias_bound}, we are ready to prove the convergence of the cost function $J(\bw_t)$ after $t_o$ iterations to {a neighborhood of} the first-order stationary point for general (geodesically $L$-smooth) g-non-convex cost functions.
\begin{theorem}
\label{theorem:convergence}
Under assumptions~\ref{assumption:graph}--\ref{assumption:gradient_noise}, suppose $\alpha \in (0, \frac{\zeta_2}{\zeta_1})$ and $\mu \in (0, \frac{1}{L}]$. The sequence $\{J(\bw_t)\}_{t\geq t_o+1}$ satisfies the following relation:
\begin{align}
    \label{eq:convergence}
    & \quad \
    \frac{1}{T-t_o}\sum_{t=t_o+1}^T\bbE \| {\nabla J}(\bw_{t})\|^2 \nonumber \\ &\leq \frac{4}{\mu K(T-t_0)} \bbE \left[J(\bw_{t_o+1}) - J(\bw^*)\right] \nonumber \\ & \quad + {\calO((\zeta_2-\zeta_1\alpha)^{-2}) + \calO(\mu)}\,,
\end{align}
where $t_o$ is given in \eqref{eq:sufficient_iteration}.
\end{theorem}
\begin{proof}
    See Appendix~\ref{appx:theorem:convergence}.
\end{proof}
Theorem~\ref{theorem:convergence} explicitly characterizes the convergence of the algorithm in~\eqref{eq:diffusion} for general geodesically smooth non-convex functions under an appropriate constant step size in a non-asymptotic manner.
{Theorem~\ref{theorem:convergence} captures the trade-off that smaller values of $\mu$ improve the asymptotic accuracy of the algorithm, while larger values of $\mu$ accelerate the transient convergence.}
{The residual does not vanish as $\alpha\to0$: a smaller combination step size $\alpha$ decreases it monotonically, down to the strictly positive floor}
{while lengthening the agreement time $t_o = \calO(\mu^{-1}\alpha^{-1}(\zeta_2-\zeta_1\alpha)^{-1})$ in Theorem~\ref{theorem:network_agreement}, after which~\eqref{eq:convergence} holds. This mirrors the classical trade-off in the gradient step size $\mu$, which balances the noise floor $\calO(\mu)$ against the transient term in~\eqref{eq:convergence}. The difference is that the residual $\calO((\zeta_2-\zeta_1\alpha)^{-2})$ decreases to a positive floor rather than to zero.}

\subsection{Convergence under the Riemannian PL condition}
{In this subsection, we discuss the convergence of the algorithm in~\eqref{eq:diffusion} when the global cost function $J$ further satisfies the {local} Riemannian PL condition in the submanifolds $\calS$. Similar to~\cite{han2025efficient, sun2024retraction}, we suppose the solutions are already close to $\calS$ after $t_p$ iterations.}
We first refine Lemma~\ref{lemma:cost_relation} using Assumption~\ref{assumption:pl_condition}. A similar refinement can be found for the Euclidean case in~\cite{xin2021improved}.
\begin{lemma}
\label{lemma:cost_relation_pl}
Under assumptions~\ref{assumption:graph}--\ref{assumption:pl_condition}, suppose $\mu \in (0, \frac{1}{L}]$ {and $\bw_t \in \calS$ after $t_p$ iterations}. The sequence {$\{J(\bw_{t})\}_{t\geq t_p+1}$} satisfies the following relation:
\begin{align}
    \label{eq:cost_relation_pl}
    \bbE \left\{J(\bw_{t+1}) - J(\bw^*)\right\} & \leq \left(1 - \frac{\mu K}{4\tau}\right)\bbE \left\{J(\bw_{t}) - J(\bw^*)\right\} \nonumber \\ & \quad  + \frac{9\alpha^2}{2\mu K}\bbE P(\bphi_{t+1}) {{}+ \mu^2L\sigma^2} \,.
\end{align}
\end{lemma}
\begin{proof}
    See Appendix~\ref{appx:lemma:cost_relation_pl}.
\end{proof}
Compared to the result in Lemma~\ref{lemma:cost_relation}, the first term on the RHS of~\eqref{eq:cost_relation_pl} in Lemma~\ref{lemma:cost_relation_pl} is reduced by a factor $(1 - \frac{\mu K}{4\tau}) < 1$ due to the PL condition, which is crucial for establishing the linear convergence of the cost function.
Combining this cost relation in~\eqref{eq:cost_relation_pl} with the bound on the consensus bias term in Lemma~\ref{lemma:consensus_bias_bound}, we are ready to prove the linear convergence of the cost function $J(\bw_t)$ after {sufficient} iterations.
\begin{theorem}
\label{theorem:convergence_pl}
Under assumptions~\ref{assumption:graph}--\ref{assumption:pl_condition}, suppose $\alpha \in (0, \frac{\zeta_2}{\zeta_1})$, $\mu \in (0, \bar \mu]$ with $\bar \mu = \min\{\frac{1}{L}, \frac{4\tau}{K}\}$ {and $\bw_t \in \calS$ after $t_p$ iterations}.
{Let $\underline t \triangleq \max\{t_o, t_p\}$,} the sequence $\{J(\bw_t)\}_{t\geq {\underline t}+1}$ satisfies the following relation:
\begin{align}
    \label{eq:convergence_pl}
    \bbE \left\{J(\bw_{t}) - J(\bw^*)\right\}
    &\leq \left(1 - \frac{\mu K}{4\tau}\right)^{t-{\underline t}}\bbE \left\{J(\bw_{{\underline t}}) - J(\bw^*)\right\} \nonumber \\ & \quad  + {\calO((\zeta_2-\zeta_1\alpha)^{-2}) + \calO(\mu)}\,.
\end{align}
\end{theorem}
\begin{proof}
    See Appendix~\ref{appx:theorem:convergence_pl}.
\end{proof}
The non-asymptotic convergence rate in Theorem~\ref{theorem:convergence_pl} indicates that the sequence $\{\bbE \left\{J(\bw_{t}) - J(\bw^*)\right\} \}_{t\geq {\underline t}+1}$ decays linearly at the rate of $(1-\frac{\mu K}{4\tau})^{t-{\underline t}}$ such that the error term is up to ${\calO((\zeta_2-\zeta_1\alpha)^{-2}) + \calO(\mu)}$ at the steady state. This suggests that the algorithm in~\eqref{eq:diffusion} can achieve a linear convergence rate after sufficient iterations ${\underline t}$ under the {local} Riemannian PL condition.
{Theorem~\ref{theorem:convergence_pl} captures the trade-off whereby larger values of $\mu$ yield a faster convergence rate, while smaller values of $\mu$ result in a smaller steady-state error.}
{The residual $\calO((\zeta_2-\zeta_1\alpha)^{-2})$ is the same as in Theorem~\ref{theorem:convergence}, and the trade-off in the combination step size $\alpha$ discussed after Theorem~\ref{theorem:convergence} applies verbatim.}

Compared to the non-convex convergence analyses in the Euclidean setting~\cite{xin2021improved, vlaski2021distributed, vlaski2021distributed2, alghunaim2022unified}, the results in Theorems~\ref{theorem:convergence} and~\ref{theorem:convergence_pl} explicitly account for the influence of manifold curvature. This curvature effect is captured by the constants $\zeta_1$, $\zeta_2$, and $C_\kappa$, which appear {explicitly in the network-disagreement residual $\calO((\zeta_2-\zeta_1\alpha)^{-2})$ of Theorems~\ref{theorem:convergence} and~\ref{theorem:convergence_pl}}. In contrast to the g-convex setting studied in~\cite{wang2025riemannian}, the convergence analysis in Theorems~\ref{theorem:convergence} and~\ref{theorem:convergence_pl} employs a Taylor expansion of the exponential map (as developed in Lemma~\ref{lemma:taylor_expansion}) to bound the consensus bias term. This approach, unlike the Lyapunov-based analysis in~\cite{wang2025riemannian}, allows the results to establish the convergence of the global cost function directly.

\section{{Example Applications}}
\label{sec:applications}
{In this section, we apply~\eqref{eq:diffusion} to distributed robust PCA on the Grassmann manifold $\GR$ and to distributed low-rank matrix completion on the fixed-rank manifold $\calM_r$.
The definitions of these two manifolds and their useful tools are presented in Appendices~\ref{appx:grassmann_manifold} and~\ref{appx:fixrank_manifold}, respectively.}

\subsection{Distributed robust PCA}
In the decentralized setting, we consider the following optimization problem inspired by~\cite{lerman2018overview, huroyan2018distributed} {with $\bx_{k} \in \bbR^{n}$ being data samples observed by each agent~$k$}:
\begin{equation}
    \label{eq:robust_pca_distributed}
	\min_{\pi(\bU_k) \in \GR}\, -\,\bbE_{\bx_k}\left\{Q_\delta(\|\bU_k^T\bx_k\|)\right\} \,,
\end{equation}
where $\pi(\bU_k)$ (see Appendix~\ref{appx:grassmann_manifold} for a definition) represents the local estimate at agent $k$, and the function $Q_\delta$ is defined as
\[
Q_\delta(p) \;=\;
\begin{cases}
p, & p \ge \delta\,, \\
\frac{p^2}{2\delta} + \frac{\delta}{2}, & p < \delta\,.
\end{cases}
\]
The expectation in the loss function~\eqref{eq:robust_pca_distributed} is approximated by realizations $\bx_{k,t}$ at each time instant $t$.

The Riemannian stochastic gradient is computed using the Euclidean gradient of~\eqref{eq:robust_pca_distributed} at $\bU_{k,t}$ and~\eqref{eq:grassmann_gradient} given in Appendix~\ref{appx:grassmann_manifold}.
The exponential mapping is defined in~\eqref{eq:grassmann_exp}.
To evaluate the accuracy of the solutions, we consider the geodesic distance~\eqref{eq:grassmann_distance} between the estimates at each time instant $\pi(\bU_{k,t})$ and the optimal solution $\pi(\bU^*)$, and we define the mean square deviation (MSD) accordingly as $\frac{1}{K}\sum_{k=1}^K d_{\GR}^2(\bU_{k,t}, \bU^*)$.
To compute $\pi(\bU^*)$ in the MSD, we use the Riemannian trust region algorithm~\cite{absil2007trust} on~\eqref{eq:robust_pca_distributed} with the full data matrix.

\begingroup
\subsection{Distributed low-rank matrix completion}
We also consider the following decentralized optimization problem on the fixed-rank manifold $\calM_r$ inspired by~\cite{vandereycken2013low}:
\begin{equation}
    \label{eq:lrmc_distributed}
    \min_{\bA_k \in \mathcal{M}_r}\, \bbE_{(i_k,\bx_{i_k})}\left\{\frac{1}{2}\|\bA_k(i_k,:) - \bx_{i_k}\|^2\right\} \,,
\end{equation}
where $\bA_k$ is the local estimate at agent $k$, and $(i_k,\bx_{i_k})$ denotes one observed row index and its corresponding row of the data matrix. The expectation in~\eqref{eq:lrmc_distributed} is approximated online by the observations available at each time instant.

The Riemannian stochastic gradient is computed as the projection of the Euclidean gradient of the instantaneous cost in~\eqref{eq:lrmc_distributed} computed according to~\eqref{eq:fixedrank_gradient}. For the combination step, we use the inverse orthographic retraction~\eqref{eq:fixedrank_inv_retraction} and the orthographic retraction~\eqref{eq:fixedrank_retraction}.
To evaluate the accuracy of the solutions, we consider the relative error (RE) used in~\cite{vandereycken2013low}, defined as $\frac{1}{K}\sum_{k=1}^K \frac{\|\bA_{k,t} - \bA^*\|_F}{\|\bA^*\|_F}$
where $\bA_{k,t}$ is the estimate at agent $k$ and iteration $t$, and $\bA^*$ is the reference low-rank solution.
To compute $\bA^*$ in the RE, we use the Riemannian conjugate gradient algorithm on~\eqref{eq:lrmc_distributed} with the full data matrix.

\endgroup

\section{{Simulation results}}
\label{sec:simulations}

\begin{figure}
    \centering
    \begin{subfigure}[b]{0.24\textwidth}
        \centering
        \includegraphics[trim=10mm 10mm 10mm 8mm, clip, scale=0.28]{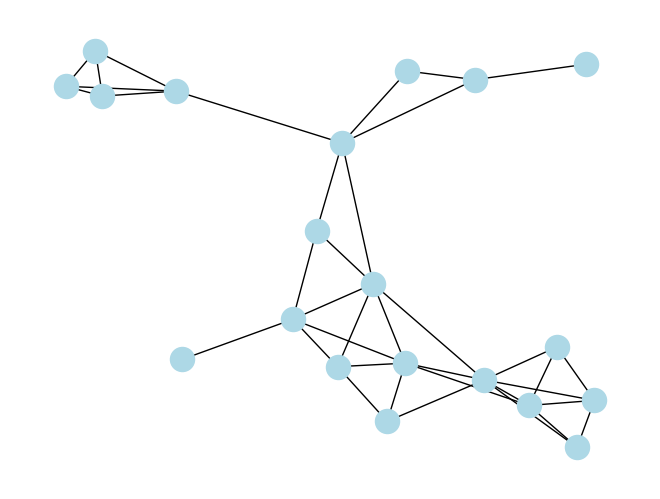}
        \caption{}
        \label{fig:topology_a}
    \end{subfigure}
    \hfill
    \begin{subfigure}[b]{0.24\textwidth}
        \centering
        \includegraphics[trim=10mm 10mm 10mm 6mm, clip, scale=0.28]{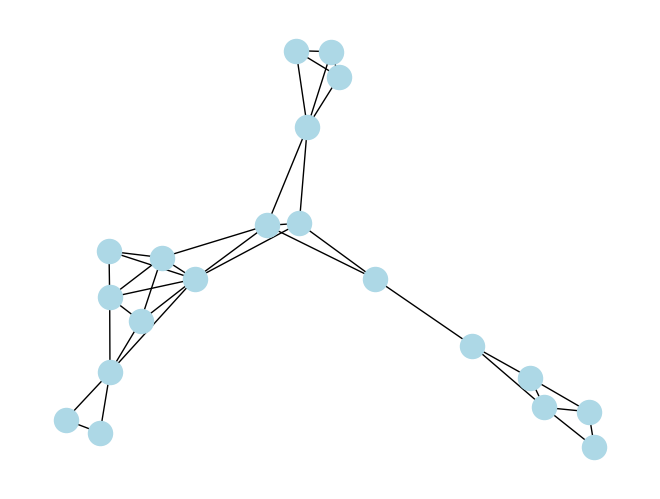}
        \caption{}
        \label{fig:topology_b}
    \end{subfigure}
    \caption{The randomly generated graph structures used in the experiments. (a) Graph with Metropolis weights. (b) Graph with uniformly distributed weights.}
    \label{fig:topology}
\end{figure}

{This section presents numerical experiments on distributed robust PCA and low-rank matrix completion.}
The randomly generated graph topologies of the multi-agent systems used for the experiments are illustrated in Figure~\ref{fig:topology}. The weights in matrix $\bC$ with $K=20$ agents were randomly generated by the Metropolis rule~\cite{sayed2014adaptationLearningOptimization} and the uniform rule.
For simulation on synthetic data, the MSD and RE curves are averaged over 100 independent Monte Carlo runs.

We compare our algorithm against the Riemannian non-cooperative algorithm, which independently applies R-SGD on each agent using its local data $\bx_{k,t}$. {For robust PCA}, we provide comparisons with an extrinsic algorithm on
Stiefel manifold: Decentralized Riemannian Stochastic Gradient Descent (DRSGD)~\cite{chen2021decentralized}.
{For low-rank matrix completion, we extend the decentralized consensus SGD~\cite{nedic2010constrained, lian2017can} to the fixed-rank manifold using a projection operator to ensure the constraints are satisfied and name the method "Extrinsic consensus" for comparison.}
{For each dataset and graph topology, we used the following two-stage tuning protocol. First, we selected the stochastic-gradient step size $\mu$ for the non-cooperative baseline by minimizing the steady-state performance metric (MSD for robust PCA and RE for matrix completion), and then reused that same gradient step size for both cooperative methods. Second, we tuned only the agreement-strength parameter of each cooperative method using the same steady-state metric.}
All the experiments are implemented in Python with the Pymanopt toolbox~\cite{townsend2016pymanopt}.
{Open-source codes to reproduce the results are publicly available on \url{https://github.com/xiuheng-wang/diffusion_manifold_nonconvex_release}.}

\subsection{Synthetic data}
For distributed robust PCA, we generate synthetic data as in~\cite{chen2021decentralized, wang2025riemannian}. First, we set $n=10$, $p=5$, and independently sample $1550K$ observations from a multivariate Gaussian model. We modify the singular values of the resulting data matrix to $0.8^i$, $i=0,\ldots,n-1$, randomly shuffle the observations, and reshape them into $K$ local streams of $1550$ observations. For each agent, we replace 50 randomly selected observations by outliers sampled uniformly from $[0,1]^n$ as in~\cite{huroyan2018distributed}; thus, each local stream contains 1500 inliers and 50 outliers.
The parameter $\delta$ in~\eqref{eq:robust_pca_distributed} is set to $0.1$.
The simulations used fixed step sizes $\mu = 0.12$ and $\alpha = 0.4$ for the Metropolis graph, and $\mu = 0.13$ and $\alpha = 0.4$ for the uniform graph.

\begingroup
For distributed low-rank matrix completion, we generate synthetic data similar to~\cite{vandereycken2013low}. First, we set $m=n=15$, $r=5$, and generate a rank-$r$ matrix with exponentially decaying singular values $0.8^{i-1}$ for $i=1,\ldots,r$. The matrix is then normalized by its Frobenius norm and perturbed by additive Gaussian noise with standard deviation $0.05$. At each time instant, each agent observes one full row of the matrix, with the row index sampled uniformly at random. The number of stochastic iterations is set to $1000$.
The simulations used fixed step sizes $\mu = 0.6$ and $\alpha = 0.6$ for both the Metropolis and uniform graphs.

Figure~\ref{fig:msd_pca} and~\ref{fig:re_lrmc} show the MSD and RE performance of the algorithms for distributed robust PCA and low-rank matrix completion on synthetic data, respectively.
It can be seen that the Riemannian diffusion adaptation strategy converges and achieves an improvement over the competing methods.
\endgroup

\begin{figure}[t]
    \centering
    \begin{subfigure}[b]{0.24\textwidth}
        \centering
        \includegraphics[trim=2mm 2mm 2mm 2mm, clip, scale=0.36]{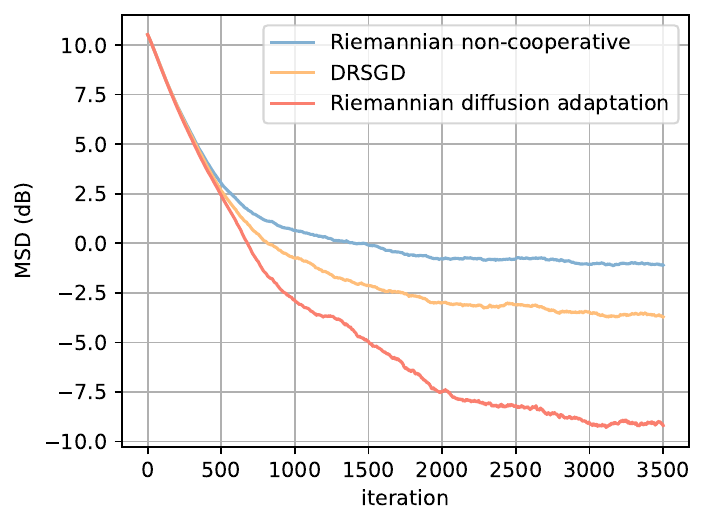}
        \caption{}
        \label{fig:msd_pca_mnist_a}
    \end{subfigure}
    \hfill
    \begin{subfigure}[b]{0.24\textwidth}
        \centering
        \includegraphics[trim=2mm 2mm 2mm 0mm, clip, scale=0.36]{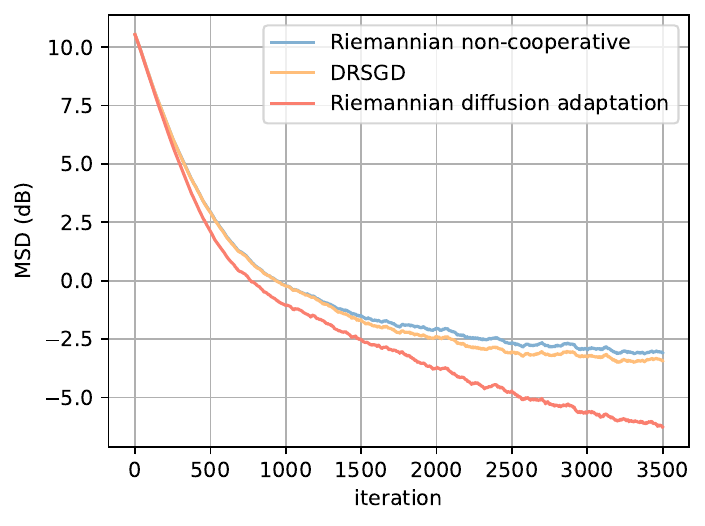}
        \caption{}
        \label{fig:msd_pca_mnist_b}
    \end{subfigure}
    \caption{Illustration of MSD performance of the algorithms for distributed robust PCA on real data on different graphs. (a) Graph with Metropolis weights. (b) Graph with uniformly distributed weights.}
    \label{fig:msd_pca_mnist}
\end{figure}

\begin{figure}[t]
    \centering
    \begin{subfigure}[b]{0.24\textwidth}
        \centering
        \includegraphics[trim=2mm 2mm 2mm 2mm, clip, scale=0.36]{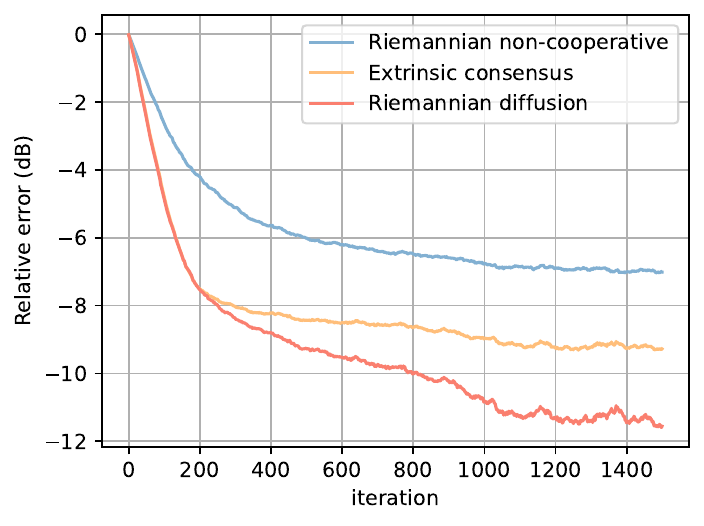}
        \caption{}
        \label{fig:re_lrmc_jester_a}
    \end{subfigure}
    \hfill
    \begin{subfigure}[b]{0.24\textwidth}
        \centering
        \includegraphics[trim=2mm 2mm 2mm 0mm, clip, scale=0.36]{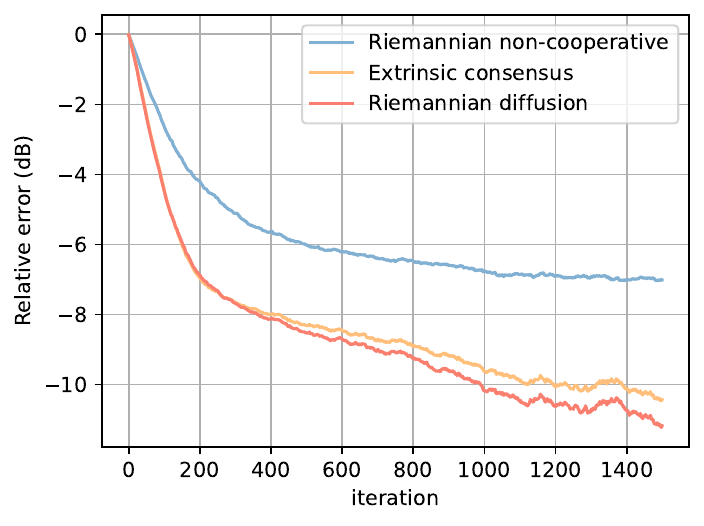}
        \caption{}
        \label{fig:re_lrmc_jester_b}
    \end{subfigure}
    \caption{{Illustration of RE performance of the algorithms for distributed low-rank matrix completion on real data on different graphs. (a) Graph with Metropolis weights. (b) Graph with uniformly distributed weights.}}
    \label{fig:re_lrmc_jester}
\end{figure}

\begin{figure}[t]
    \centering
    \begin{subfigure}[b]{0.24\textwidth}
        \centering
        \includegraphics[trim=2mm 2mm 2mm 2mm, clip, scale=0.36]{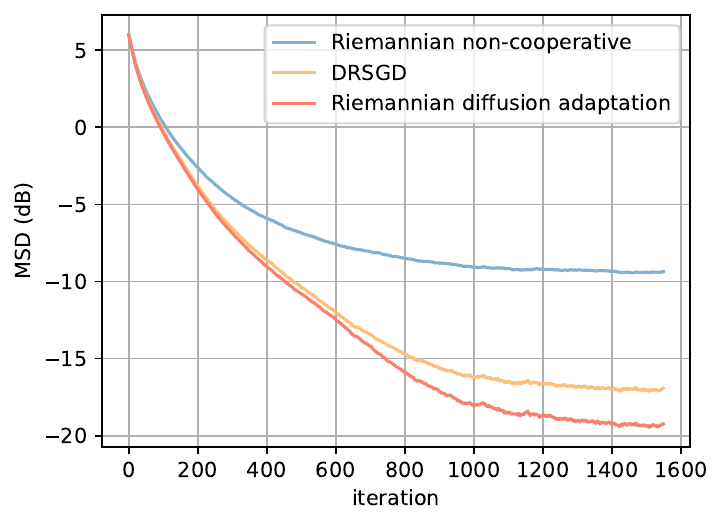}
        \caption{}
        \label{fig:msd_pca_a}
    \end{subfigure}
    \hfill
    \begin{subfigure}[b]{0.24\textwidth}
        \centering
        \includegraphics[trim=2mm 2mm 2mm 0mm, clip, scale=0.36]{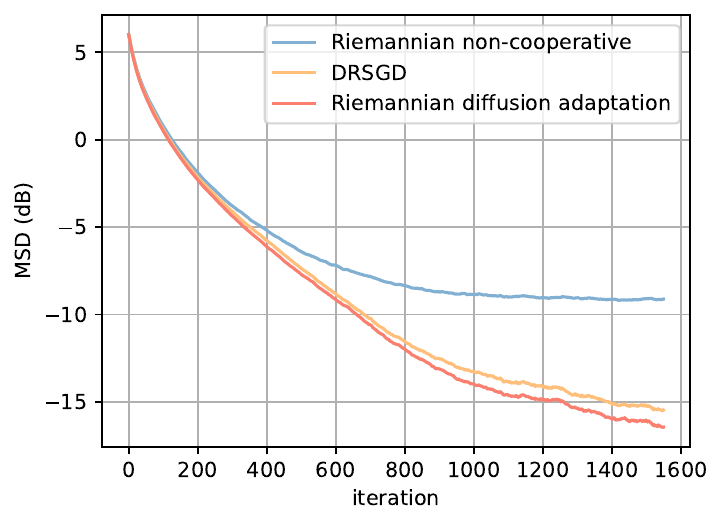}
        \caption{}
        \label{fig:msd_pca_b}
    \end{subfigure}
    \caption{Illustration of MSD performance of the algorithms for distributed robust PCA on synthetic data on different graphs. (a) Graph with Metropolis weights. (b) Graph with uniformly distributed weights.}
    \label{fig:msd_pca}
\end{figure}

\begin{figure}[t]
    \centering
    \begin{subfigure}[b]{0.24\textwidth}
        \centering
        \includegraphics[trim=2mm 2mm 2mm 2mm, clip, scale=0.36]{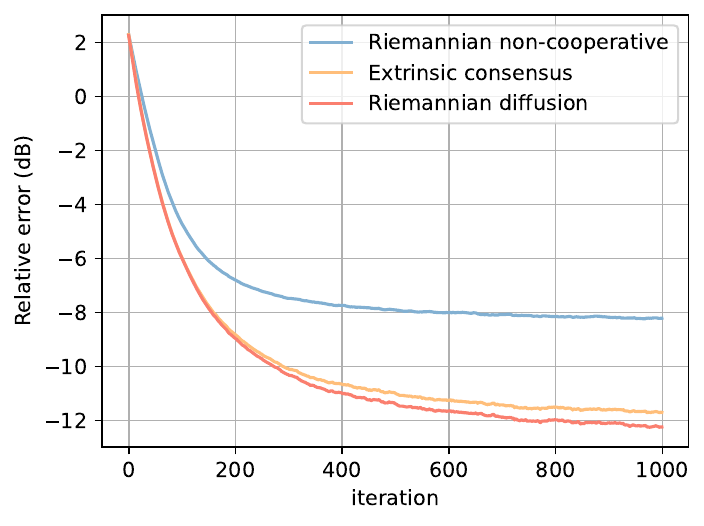}
        \caption{}
        \label{fig:re_lrmc_a}
    \end{subfigure}
    \hfill
    \begin{subfigure}[b]{0.24\textwidth}
        \centering
        \includegraphics[trim=2mm 2mm 2mm 0mm, clip, scale=0.36]{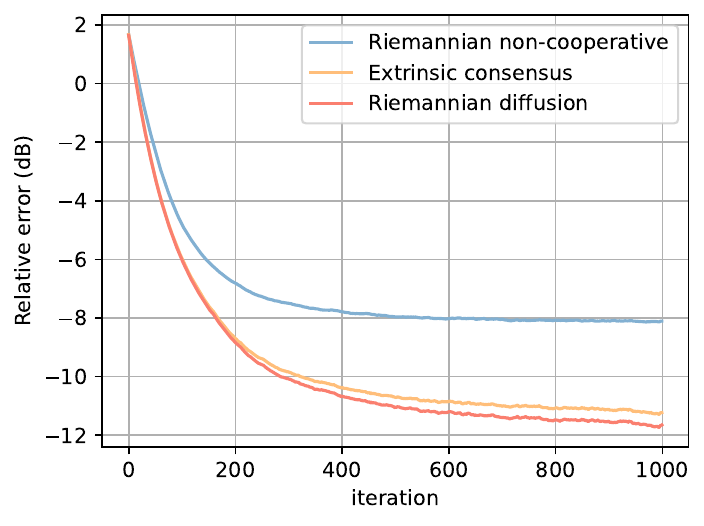}
        \caption{}
        \label{fig:re_lrmc_b}
    \end{subfigure}
    \caption{{Illustration of RE performance of the algorithms for distributed low-rank matrix completion on synthetic data on different graphs. (a) Graph with Metropolis weights. (b) Graph with uniformly distributed weights.}}
    \label{fig:re_lrmc}
\end{figure}

\subsection{Real data}
{For distributed robust PCA,} we obtain numerical results on the MNIST dataset~\cite{lecun1998mnist}, which contains 70000 hand-written images with $n = 784$ pixels. The data matrix is normalized such that the elements are in the range $[0,1]$ and then centered.
We randomly shuffle the images, partition them into $K=20$ subsets, and then run the algorithms to compute the first $p = 5$ principal components.
The step sizes are set to $\mu = 0.006$ and $\alpha = 0.005$ for the Metropolis graph and $\mu = 0.006$ and $\alpha = 0.001$ for the uniform graph.

\begingroup
For distributed low-rank matrix completion, we consider the Jester dataset~\cite{goldberg2001eigentaste}. We randomly select 50 users who rated at least 50 jokes, rescale the observed ratings to $[-1,1]$, and fill each missing entry with the corresponding user's mean observed rating. At each stochastic iteration, every agent independently samples one user row from this common imputed matrix. We run 1500 iterations to recover a rank-$r$ approximation with $r=20$. The step sizes are set to $\mu=0.8$ and $\alpha=0.3$ for both the Metropolis and uniform graphs.

The MSD and RE of the different methods across both graphs, shown in Figure~\ref{fig:msd_pca_mnist} and~\ref{fig:re_lrmc_jester}, behave similarly to those in the experiment with synthetic data, showing similar convergence and comparable performance across the different approaches.

The improved MSD performance of Riemannian diffusion adaptation on both synthetic and real data can be attributed to the enhanced stability of diffusion strategies compared to consensus counterparts under constant step sizes~\cite{sayed2014adaptationLearningOptimization,chen2015learning}, and its intrinsic manifold updates, which avoid extrinsic distortions during adaptation and combination.
\endgroup

\section{Conclusion}
\label{sec:conclusion}
In this work, the Riemannian diffusion adaptation algorithm was studied for decentralized optimization over multi-agent networks in geodesically non-convex environments.
We showed that the iterates of the agents achieve approximate consensus after a sufficient number of iterations.
Building on this finding, we established the convergence of the algorithm to {a neighborhood of} a stationary point for general geodesically non-convex costs and linear convergence under the additional {local} Riemannian PL condition. {The results showed that the algorithm can achieve a small steady-state error bound by choosing a small step size $\mu$, up to a residual.}
We applied the algorithm to the online distributed robust PCA {and low-rank matrix completion problems}. Numerical simulations on both synthetic and real data illustrated the convergence and performance of the algorithm.

\appendices
\section{Proof of lemma~\ref{lemma:frechet_variance_descent_one_iteration}}
\label{appx:lemma:frechet_variance_descent_one_iteration}
\begin{proof}
Define $\bphi_{m,t}$ as the Fréchet mean of $\bphi_{t}$, we can apply the inequality~\eqref{eq:trigonometric_upper_bound} in Lemma~\ref{lemma:trigonometric_distance} to the geodesic triangle $\Delta \bw_{k, t}\bphi_{k, t}\bphi_{m,t}$ and obtain that
\begin{align}
    \label{eq:trigonometric_upper_bound_1}
    d^2(\bw_{k, t}, \bphi_{m,t}) \leq \ & \zeta_1 d^2(\bphi_{k, t}, \bw_{k, t}) + d^2(\bphi_{k, t}, \bphi_{m, t}) \nonumber \\ & - 2\langle \exp_{\bphi_{k, t}}^{-1}(\bw_{k, t}),  \exp_{\bphi_{k, t}}^{-1}(\bphi_{m, t}) \rangle\,.
\end{align}
From the combination step in~\eqref{eq:diffusion}, we know
\begin{align}
    \label{eq:opt_penalty_1}
    \exp_{\bphi_{k, t}}^{-1}(\bw_{k, t}) = \alpha \sum_{\ell=1}^K c_{\ell k} \exp^{-1}_{\bphi_{k, t}}(\bphi_{\ell, t})\,,
\end{align}
so that
\begin{align}
    \label{eq:opt_penalty_2}
    & \ 2\langle \exp_{\bphi_{k, t}}^{-1}(\bw_{k, t}),  \exp_{\bphi_{k, t}}^{-1}(\bphi_{m, t}) \rangle \nonumber \\ = & \ 2\alpha \sum_{\ell=1}^K c_{\ell k} \langle \exp^{-1}_{\bphi_{k, t}}(\bphi_{\ell, t}),  \exp_{\bphi_{k, t}}^{-1}(\bphi_{m, t}) \rangle \,.
\end{align}
Now we lower bound the item on the right-hand side (RHS) of~\eqref{eq:opt_penalty_2} by applying the inequality~\eqref{eq:trigonometric_lower_bound} in Lemma~\ref{lemma:trigonometric_distance} to the geodesic triangle $\Delta \bphi_{\ell, t}\bphi_{k, t}\bphi_{m,t}$, and thus obtain
\begin{align}
    \label{eq:trigonometric_lower_bound_1}
    d^2(\bphi_{\ell, t}, \bphi_{m,t}) & \geq \zeta_2 d^2(\bphi_{k, t}, \bphi_{\ell,t}) + d^2(\bphi_{k, t}, \bphi_{m,t}) \nonumber \\ & \quad - 2\langle \exp_{\bphi_{k, t}}^{-1}(\bphi_{\ell, t}),  \exp_{\bphi_{k, t}}^{-1}(\bphi_{m, t})\rangle\,.
\end{align}
Combine the results in~\eqref{eq:trigonometric_upper_bound_1}, ~\eqref{eq:opt_penalty_2} and~\eqref{eq:trigonometric_lower_bound_1}, we have
\begin{align}
    \label{eq:trigonometric_bound_combined}
    d^2(\bw_{k, t}, \bphi_{m,t}) & \leq \zeta_1 d^2(\bphi_{k, t}, \bw_{k, t}) + d^2(\bphi_{k, t}, \bphi_{m, t}) \nonumber \\ & \quad - \zeta_2 \alpha \sum_{\ell=1}^K c_{\ell k} d^2(\bphi_{k, t}, \bphi_{\ell,t}) \nonumber \\ & \quad + \alpha \sum_{\ell=1}^K c_{\ell k} d^2(\bphi_{\ell, t}, \bphi_{m,t}) \nonumber \\ & \quad - \alpha \sum_{\ell=1}^K c_{\ell k} d^2(\bphi_{k, t}, \bphi_{m,t}) \,.
\end{align}
Summing~\eqref{eq:trigonometric_bound_combined} over $k$ and consider $C$ is doubly stochastic as in Assumption~\ref{assumption:graph}, we have
\begin{align}
    \label{eq:trigonometric_bound_combined_1}
    \sum_{k=1}^Kd^2(\bw_{k, t}, \bphi_{m,t}) & \leq \zeta_1 \sum_{k=1}^K d^2(\bphi_{k, t}, \bw_{k, t}) + V_F(\bphi_t) \nonumber \\
    & \quad - \zeta_2 \alpha P(\bphi_t) \,.
\end{align}
To further upper bound the RHS of~\eqref{eq:trigonometric_bound_combined_1}, from~\eqref{eq:opt_penalty_1} we write
\begin{align}
    \label{eq:opt_penalty_3}
    d^2(\bphi_{k, t}, \bw_{k, t}) &= \alpha^2 \bigg\|\sum_{\ell=1}^K c_{\ell k} \exp^{-1}_{\bphi_{k, t}}(\bphi_{\ell, t})\bigg\|^2 \nonumber \\
    & \leq \alpha^2 \sum_{\ell=1}^K c_{\ell k}\sum_{\ell=1}^K c_{\ell k} \big\| \exp^{-1}_{\bphi_{k, t}}(\bphi_{\ell, t})\big\|^2 \nonumber \\
    & = \alpha^2 P(\bphi_t) \,,
\end{align}
where the first inequality uses the Cauchy-Schwarz inequality and the second equality follows from the fact that $C$ is doubly stochastic. Define $\bw_{m,t}$ as the Fréchet mean of $\bw_{t}$,
we can plug the result in~\eqref{eq:opt_penalty_3} into~\eqref{eq:trigonometric_bound_combined_1} to obtain
\begin{align}
\label{eq:trigonometric_bound_combined_2}
    V_F(\bw_t)  &\leq \sum_{k=1}^Kd^2(\bw_{k, t}, \bphi_{m,t}) \nonumber \\
    & \leq V_F(\bphi_t) + (\zeta_1\alpha^2 - \zeta_2\alpha) P(\bphi_t)  \,,
\end{align}
as desired.
\end{proof}

\section{Proof of lemma~\ref{lemma:graph_topology}}
\label{appx:lemma:graph_topology}
\begin{proof}

Apply Lemma~\ref{lemma:lipschitz_exponential_map} to the variables $\bphi_{k, t}, \bphi_{\ell, t}$ and $\bphi_{m, t}$, we have
\begin{align}
    \label{eq:lipschitz_exponential_map_0}
    & (1+C_\kappa B^2)^2 d^2(\bphi_{k, t}, \bphi_{\ell, t}) \nonumber \\
    \geq \ & \|\exp^{-1}_{\bphi_{m, t}} (\bphi_{k, t}) - \exp^{-1}_{\bphi_{m, t}} (\bphi_{\ell, t})\|^2 \,.
\end{align}
Multiply~\eqref{eq:lipschitz_exponential_map_0} by $c_{\ell, k}$ and summarize the result over $\ell$ and $k$, from the symmetric and doubly stochastic properties of $C$ in Assumption~\ref{assumption:graph} we can obtain the following result:
\begin{align}
    \label{eq:lipschitz_exponential_map_1}
    & (1+C_\kappa B^2)^2 P(\bphi_t)  \nonumber \\
    \geq \ & \sum_{k=1}^K\sum_{\ell=1}^K c_{\ell k}\|\exp^{-1}_{\bphi_{m, t}} (\bphi_{k, t}) - \exp^{-1}_{\bphi_{m, t}} (\bphi_{\ell, t})\|^2 \nonumber \\
    = \ & 2V_F(\bphi_t) - 2\sum_{k=1}^K\sum_{\ell=1}^K c_{\ell k}\langle \exp_{\bphi_{m, t}}^{-1}(\bphi_{k, t}),  \exp_{\bphi_{m, t}}^{-1}(\bphi_{\ell, t})\rangle\,.
\end{align}
Consider that $\sum_{\ell=1}^{K} \exp_{\bphi_{m, t}}^{-1}(\bphi_{\ell, t}) = \bzero$ since $\bphi_{m, t}$ is the Fréchet mean of $\{\bphi_{1, t}, \cdots, \bphi_{K, t}\}$, we can write:
\begin{align}
    \label{eq:graph_topology_proof}
    &\sum_{k=1}^K\sum_{\ell=1}^K c_{\ell k}\langle \exp_{\bphi_{m, t}}^{-1}(\bphi_{k, t}),  \exp_{\bphi_{m, t}}^{-1}(\bphi_{\ell, t})\rangle \nonumber \\ =\ &\sum_{k=1}^K\sum_{\ell=1}^K \left(c_{\ell k} - \frac{1}{K}\right)\langle \exp_{\bphi_{m, t}}^{-1}(\bphi_{k, t}),  \exp_{\bphi_{m, t}}^{-1}(\bphi_{\ell, t})\rangle\,.
\end{align}
By selecting an appropriate orthonormal basis, we can represent the elements $\{\exp_{\bphi_{m, t}}^{-1}(\bphi_{k, t})\}$ as matrix notation $\bU \triangleq \left[\exp_{\bphi_{m, t}}^{-1}(\bphi_{1, t}), \cdots, \exp_{\bphi_{m, t}}^{-1}(\bphi_{K, t})\right]$ in the tangent space $T_{\bphi_{m, t}}\calM^K$. Thus, we can further write the inner in matrix form and consider the symmetric property of $C$ and the fact $\sum_{\ell=1}^{K} \exp_{\bphi_{m, t}}^{-1}(\bphi_{\ell, t}) = \bzero$ to obtain:
\begin{align}
    \label{eq:graph_topology_proof_1}
    &\sum_{k=1}^K\sum_{\ell=1}^K \left(c_{\ell k} - \frac{1}{K}\right)\langle \exp_{\bphi_{m, t}}^{-1}(\bphi_{k, t}),  \exp_{\bphi_{m, t}}^{-1}(\bphi_{\ell, t})\rangle  \nonumber\\ = \ & \tr\left(\bU\left(C - \frac{1}{K}\bone\bone^T\right)\bU^T\right)\,,
\end{align}
where the matrix $C - \frac{1}{K}\bone\bone^T$ represents the deviation from consensus.
From the definition of spectral radius, we can further bound the above trace as follows:
\begin{align}
    \label{eq:graph_topology_proof_2}
    \tr\left(\bU\left(C - \frac{1}{K}\bone\bone^T\right)\bU^T\right) &\leq \lambda \tr\left(\bU\bU^T\right) \nonumber \\ &= \lambda V_F(\bphi_t)\,,
\end{align}
with the mixing rate of the network $\lambda = \rho(C-\frac{1}{K}\bone\bone^T)$ defined in Lemma~\ref{lemma:graph_mixing_rate}.
Combining the results in~\eqref{eq:lipschitz_exponential_map_1}-\eqref{eq:graph_topology_proof_2}, we obtain:
\begin{align}
    \label{eq:graph_topology_proof_3}
    P(\bphi_t) \geq \frac{2(1-\lambda)}{(1+C_\kappa B^2)^2} V_F(\bphi_t)\,,
\end{align}
Combining the result in Lemma~\ref{lemma:frechet_variance_descent_one_iteration} and~\eqref{eq:graph_topology_proof_3}, we obtain the desired result.
{Note that $\varepsilon \in (0, \frac{1}{2}]$ for every $\alpha \in (0, \frac{\zeta_2}{\zeta_1})$, since $\zeta_2\alpha - \zeta_1\alpha^2 \leq \frac{\zeta_2^2}{4\zeta_1}$, $0 \leq \lambda < 1$, $0 < \zeta_2 \leq 1 \leq \zeta_1$, and $C_\kappa B^2 \geq 0$.}
\end{proof}

\section{Proof of lemma~\ref{lemma:frechet_variance_descent}}
\label{appx:lemma:frechet_variance_descent}
\begin{proof}
From Lemma~\ref{lemma:graph_topology} and the fact that $\bphi_{m,t}$ is the Fréchet mean of $\bphi_{t}$, we have
\begin{align}
\label{eq:trigonometric_bound_combined_4}
    V_F(\bw_{t})  \leq \left(1 - \varepsilon \right) \sum_{k=1}^K d^2(\bphi_{k, t}, \bw_{m, t-1}) \,.
\end{align}
To further upper bound the term $d^2(\bphi_{k, t}, \bw_{m, t-1})$,
we apply the inequality~\eqref{eq:trigonometric_upper_bound} in Lemma~\ref{lemma:trigonometric_distance} to the geodesic triangle $\Delta \bphi_{k, t}\bw_{k, t-1}\bw_{m,t-1}$ and obtain that
\begin{align}
    \label{eq:trigonometric_upper_bound_2}
    & \quad \ d^2(\bphi_{k, t}, \bw_{m, t-1}) \nonumber \\ & \leq  \zeta_1 d^2(\bw_{k, t-1}, \bphi_{k, t}) + d^2(\bw_{k, t-1}, \bw_{m, t-1}) \nonumber \\ & \quad - 2\langle \exp_{\bw_{k, t-1}}^{-1}(\bphi_{k, t}),  \exp_{\bw_{k, t-1}}^{-1}(\bw_{m, t-1}) \rangle \nonumber \\
    & = \zeta_1 \mu^2 \|\widehat{\nabla J}_k(\bw_{k, t-1})\|^2  + d^2(\bw_{k, t-1}, \bw_{m, t-1}) \nonumber \\ & \quad + 2 \langle \mu\widehat{\nabla J}_k(\bw_{k, t-1}), \exp_{\bw_{k, t-1}}^{-1}(\bw_{m, t-1})\rangle \,,
\end{align}
where the equality follows that $\exp_{\bw_{k, t-1}}^{-1}(\bphi_{k, t}) = - \mu \widehat{\nabla J}_k(\bw_{k, t-1})$ from the adaptation step in~\eqref{eq:diffusion}. Take expectation on~\eqref{eq:trigonometric_upper_bound_2} w.r.t.
$\{\bx_{k, s}\}_{s=0}^t$, we have
\begin{align}
    \label{eq:trigonometric_upper_bound_3}
    &\quad \ \bbE d^2(\bphi_{k, t}, \bw_{m, t-1}) \nonumber \\
    & \leq \zeta_1 \mu^2 \bbE\|\widehat{\nabla J}_k(\bw_{k, t-1})\|^2 + \bbE d^2(\bw_{k, t-1}, \bw_{m, t-1})  \nonumber \\ & \quad + 2 \bbE\langle \mu\bbE\{\widehat{\nabla J}_k(\bw_{k, t-1})|\calF_{t-1}\}, \exp_{\bw_{k, t-1}}^{-1}(\bw_{m, t-1}) \rangle \nonumber \\
    & = \zeta_1 \mu^2 \bbE\|\widehat{\nabla J}_k(\bw_{k, t-1})\|^2 + \bbE d^2(\bw_{k, t-1}, \bw_{m, t-1}) \nonumber \\ & \quad + 2 \bbE\langle \mu\nabla J_k(\bw_{k, t-1}), \exp_{\bw_{k, t-1}}^{-1}(\bw_{m, t-1}) \rangle \nonumber \\
    & \leq \zeta_1 \mu^2 \bbE\|\widehat{\nabla J}_k(\bw_{k, t-1})\|^2 + (1+\xi)\bbE d^2(\bw_{k, t-1}, \bw_{m, t-1})\nonumber \\ & \quad + \xi^{-1}\mu^2 \bbE\|\nabla J_k(\bw_{k, t-1})\|^2 \nonumber \\
    & \leq 2\zeta_1 \mu^2 \left(\bbE\|\nabla J_k(\bw_{k, t-1})\|^2 + \bbE\{\|\bs_{k, t}\|^2|\calF_{t-1}\}\right) \nonumber \\ & \quad + (1+\xi)\bbE d^2(\bw_{k, t-1}, \bw_{m, t-1}) + \xi^{-1}\mu^2 \bbE\|\nabla J_k(\bw_{k, t-1})\|^2 \nonumber \\
    & = (1+\xi)\bbE d^2(\bw_{k, t-1}, \bw_{m, t-1}) \nonumber \\ & \quad + \mu^2\left(2\zeta_1 G^2 + \xi^{-1} G^2 + 2\zeta_1 \sigma^2\right)\,,
\end{align}
where we use the facts $2\langle a, b \rangle \leq \xi a^2 + \xi^{-1}b^2$ for $\xi > 0$ and $\frac{1}{2} (a+b)^2 \leq a^2 + b^2$ in the second and third inequalities, respectively, and consider Assumption~\ref{assumption:gradient_noise} and Lemma~\ref{lemma:gradient_bound} in the equalities. Take expectation on~\eqref{eq:trigonometric_bound_combined_4} w.r.t.
$\{\bx_{k, s}\}_{s=0}^t$ and combine the result with~\eqref{eq:trigonometric_upper_bound_3} to obtain, for every $\xi > 0$,
\begingroup
\begin{align}
\label{eq:frechet_recursion_xi}
    \bbE V_F(\bw_{t}) \leq{}& (1-\varepsilon)(1+\xi)\, \bbE V_F(\bw_{t-1}) \nonumber \\ & + (1-\varepsilon)\mu^2 K \left(2\zeta_1 G^2 + \xi^{-1} G^2 + 2\zeta_1 \sigma^2\right)\,.
\end{align}
Recalling that $\varepsilon \in (0, \frac{1}{2}]$ (see Lemma~\ref{lemma:graph_topology}), we select
\begin{align}
\label{eq:xi_choice}
    \xi = \frac{\varepsilon}{2(1-\varepsilon)} > 0\,,
\end{align}
for which $(1-\varepsilon)(1+\xi)  = 1 - \frac{\varepsilon}{2}$ and $\xi^{-1} = 2(1-\varepsilon)\varepsilon^{-1}$. Substituting these values into~\eqref{eq:frechet_recursion_xi} yields the desired result.
\endgroup
\end{proof}

\section{Proof of Theorem~\ref{theorem:network_agreement}}
\label{appx:theorem:network_agreement}
\begin{proof}
We can iterate the result in Lemma~\ref{lemma:frechet_variance_descent}, starting from $t=0$ to obtain
\begin{align}
    \label{eq:trigonometric_bound_combined_5}
    & \quad \ \bbE V_F(\bw_{t}) \nonumber \\ &\leq (1 - {\varepsilon/2})^t V_F(\bw_{k, 0}) \nonumber \\ & \quad + \left(1 - \varepsilon\right) \mu^2K\big(2\zeta_1 G^2 + {2(1-\varepsilon)}\varepsilon^{-1} G^2 \nonumber \\ & \qquad + 2\zeta_1 \sigma^2\big) \sum_{s=0}^{t-1} (1 - {\varepsilon/2})^s \nonumber \\
    &\leq (1 - {\varepsilon/2})^t KB^2 \nonumber \\ & \quad + {2}\left(1 - \varepsilon\right) \varepsilon^{-{1}} \mu^2K\big(2\zeta_1 G^2 + {2(1-\varepsilon)}\varepsilon^{-1} G^2 \nonumber \\ & \qquad + 2\zeta_1 \sigma^2\big) \nonumber \\
    &\leq {4}\left(1 - \varepsilon\right) \varepsilon^{-{1}} \mu^2K\big(2\zeta_1 G^2 + {2(1-\varepsilon)}\varepsilon^{-1} G^2 \nonumber \\ & \qquad + 2\zeta_1 \sigma^2\big) \,,
\end{align}
where the second inequality follows from the facts $V_F(\bw_{k, 0}) = \sum_{k=1}^Kd^2(\bw_{k, 0}, \bw_{m,0}) \leq KB^2$ and $\sum_{s=0}^{t-1} (1 - {\varepsilon/2})^s \leq \sum_{s=0}^{\infty} (1 - {\varepsilon/2})^s \leq {2}\varepsilon^{-{1}}$. The last inequality holds whenever
\begin{align}
    &(1 - {\varepsilon/2})^t KB^2 \nonumber \\ & \quad \leq {2}\left(1 - \varepsilon\right) \varepsilon^{-{1}} \mu^2K\left(2\zeta_1 G^2 + {2(1-\varepsilon)}\varepsilon^{-1} G^2 + 2\zeta_1 \sigma^2\right) \nonumber \\
    & \Longleftrightarrow  t \log(1 - {\varepsilon/2}) \leq 2\log(\mu) + \calO(1) \nonumber \\
    & \Longleftrightarrow t \geq \frac{2\log(\mu)}{\log(1 - {\varepsilon/2})} + \calO(1)\,.
\end{align}
We conclude that
\begin{align}
    \bbE V_F(\bw_{t}) & \leq {4}\left(1 - \varepsilon\right) \varepsilon^{-{1}} \mu^2{K}\big(2\zeta_1 G^2 + {2(1-\varepsilon)}\varepsilon^{-1} G^2 \nonumber \\ & \qquad {}+ 2\zeta_1 \sigma^2\big) \nonumber \\ & = {\calO(\mu^2\alpha^{-2}(\zeta_2-\zeta_1\alpha)^{-2})}\,,
\end{align}
with small step sizes $\mu$ after sufficient iterations $t_o$, where
\begin{align}
    t_o =  \frac{2\log(\mu)}{\log(1 - {\varepsilon/2})} + \calO(1) = {\calO(\mu^{-1}\alpha^{-1}(\zeta_2-\zeta_1\alpha)^{-1})}\,,
\end{align}
{where the second equality follows since $-\log(1-x) \geq x$ for $x \in (0,1)$ implies $-\log(1 - \frac{\varepsilon}{2}) \geq \frac{\varepsilon}{2}$, so that the leading term of $t_o$ is at most $4\varepsilon^{-1}\log(1/\mu)$, i.e., $t_o = \calO(\varepsilon^{-1}\log(1/\mu))$; the upper bound $\calO(\mu^{-1}\varepsilon^{-1})$ then follows from the fact $\log(1/\mu) \leq \mu^{-1}$ for $\mu \in (0,1)$. Finally, $\varepsilon^{-1} = \calO(\alpha^{-1}(\zeta_2-\zeta_1\alpha)^{-1})$ by~\eqref{eq:varepsilon}, which gives the desired result.}
\end{proof}

\section{Proof of lemma~\ref{lemma:cost_relation}}
\label{appx:lemma:cost_relation}
\begin{proof}
Considering the smoothness of $J_k$ in Assumption~\ref{assumption:smooth} with $\exp_{\bw_{k, t}}^{-1}(\bphi_{k, t+1}) = - \mu \widehat{\nabla J_k}(\bw_{k, t})$ from the adaptation step in~\eqref{eq:diffusion}, we can write:
\begin{align}
\label{eq:smooth_cost}
    J_k(\bphi_{k, t+1}) &\leq J_k(\bw_{k, t}) + \langle \nabla J_k(\bw_{k, t}), \exp_{\bw_{k, t}}^{-1}(\bphi_{k, t+1})\rangle \nonumber \\ & \quad + \frac{L\|\exp_{\bw_{k, t}}^{-1}(\bphi_{k, t+1})\|^2}{2} \nonumber\\
    & = J_k(\bw_{k, t}) + \langle \nabla J_k(\bw_{k, t}),  - \mu \widehat{\nabla J_k}(\bw_{k, t})\rangle \nonumber \\ & \quad + \frac{L\| - \mu \widehat{\nabla J_k}(\bw_{k, t})\|^2}{2} \,.
\end{align}
Also, we can obtain the following bound from the geodesic smoothness of $J_k$:
\begin{align}
\label{eq:smooth_grad_cost}
\|\nabla J_k(\bw_{k, t})  - \Gamma_{\bphi_{k, t+1}}^{\bw_{k, t}} \nabla J_k(\bphi_{k, t+1})\| \leq L\mu\|\widehat{\nabla J}_k(\bw_{k, t})\|\,.
\end{align}
Taking the expectation on~\eqref{eq:smooth_cost} w.r.t. $\{\bx_{k, s}\}_{s=0}^t$ and considering~\eqref{eq:gradient_mean} in Assumption~\ref{assumption:gradient_noise}, we have for each agent $k$:
\begingroup
\begin{align}
\label{eq:smooth_cost_expectation}
    &\quad \bbE J_k(\bphi_{k, t+1}) \nonumber \\ &\leq \bbE J_k(\bw_{k, t}) + \bbE\{\langle \nabla J_k(\bw_{k, t}),  - \mu \widehat{\nabla J}_k(\bw_{k, t})\rangle \} \nonumber \\ & \quad + \frac{L\bbE \| - \mu \widehat{\nabla J}_k(\bw_{k, t})\|^2}{2} \nonumber \\
    & = \bbE J_k(\bw_{k, t}) - \mu\bbE \|\nabla J_k(\bw_{k, t})\|^2 \nonumber \\ & \quad + \frac{\mu^2L}{2}\left(\bbE \|\nabla J_k(\bw_{k, t})\|^2 + \bbE \|\bs_{k, t+1}(\bw_{k, t})\|^2\right) \nonumber \\
    & \leq \bbE J_k(\bw_{k, t}) - \epsilon \bbE \|\nabla J_k(\bw_{k, t})\|^2 + \frac{\mu^2L}{2}\sigma^2 \,,
\end{align}
where $\epsilon \triangleq \mu\Big( 1 - \frac{\mu L}{2}\Big) > 0$ since $\mu \in (0, \frac{1}{L}]$. The equality uses the unbiasedness~\eqref{eq:gradient_mean}, i.e., $\bbE\{\widehat{\nabla J}_k(\bw_{k, t})|\calF_t\} = \nabla J_k(\bw_{k, t})$, together with $\bbE \| \widehat{\nabla J}_k(\bw_{k, t})\|^2 = \bbE \|\nabla J_k(\bw_{k, t})\|^2 + \bbE \|\bs_{k, t+1}(\bw_{k, t})\|^2$, and the last inequality follows from the bound on the gradient noise~\eqref{eq:gradient_variance_square}.
\endgroup
Again, considering the smoothness of $J_k$ in Assumption~\ref{assumption:smooth} with $\exp^{-1}_{\bphi_{k, t+1}}(\bw_{k, t+1}) = \alpha \sum_{\ell=1}^K c_{\ell k} \exp^{-1}_{\bphi_{k, t+1}}(\bphi_{\ell, t+1})$ from the combination step in~\eqref{eq:diffusion}, we obtain:
\begin{align}
\label{eq:smooth_penalty}
    & \quad \ J_k(\bw_{k, t+1}) \nonumber \\
    &\leq J_k(\bphi_{k, t+1}) + \langle \nabla J_k(\bphi_{k, t+1}), \exp^{-1}_{\bphi_{k, t+1}}(\bw_{k, t+1})\rangle \nonumber \\ & \quad + \frac{L\|\exp^{-1}_{\bphi_{k, t+1}}(\bw_{k, t+1})\|^2}{2} \nonumber\\
    & = J_k(\bphi_{k, t+1}) + \langle \nabla J_k(\bphi_{k, t+1}),  \alpha \sum_{\ell=1}^K c_{\ell k} \exp^{-1}_{\bphi_{k, t+1}}(\bphi_{\ell, t+1})\rangle \nonumber \\ & \quad + \frac{L\| \alpha \sum_{\ell=1}^K c_{\ell k} \exp^{-1}_{\bphi_{k, t+1}}(\bphi_{\ell, t+1})\|^2}{2} \nonumber \\
    & \leq J_k(\bphi_{k, t+1}) + \frac{\xi}{2}\|\nabla J_k(\bphi_{k, t+1})\|^2 \nonumber \\ & \quad + \left(\frac{1}{2\xi}  + \frac{L}{2}\right)\alpha^2 \bigg\| \sum_{\ell=1}^K c_{\ell k} \exp^{-1}_{\bphi_{k, t+1}}(\bphi_{\ell, t+1})\bigg\|^2\,,
\end{align}
where the second inequality uses the fact $\langle a, b \rangle \leq \frac{\xi}{2} a^2 + \frac{1}{2\xi}b^2$. Then, we take the expectation on~\eqref{eq:smooth_penalty} w.r.t. $\{\bx_{k, s}\}_{s=0}^t$, and combine the result with~\eqref{eq:smooth_cost_expectation} to obtain
\begingroup
\begin{align}
\label{eq:smooth_both}
    & \quad \ \bbE J_k(\bw_{k, t+1}) \nonumber \\ &\leq \bbE J_k(\bw_{k, t}) - \epsilon \bbE \| {\nabla J}_k(\bw_{k, t})\|^2 +  \frac{\xi}{2}\bbE\|\nabla J_k(\bphi_{k, t+1})\|^2 \nonumber \\ & \quad + \left(\frac{1}{2\xi}  + \frac{L}{2}\right)\alpha^2 \bbE\bigg\|\sum_{\ell=1}^K c_{\ell k} \exp^{-1}_{\bphi_{k, t+1}}(\bphi_{\ell, t+1})\bigg\|^2 + \frac{\mu^2L}{2}\sigma^2 \nonumber \\
    & \leq \bbE J_k(\bw_{k, t}) - \epsilon \bbE \| {\nabla J}_k(\bw_{k, t})\|^2  + \frac{\xi}{2}\bbE\|\nabla J_k(\bphi_{k, t+1})\|^2  \nonumber \\ & \quad + \left(\frac{1}{2\xi}  + \frac{L}{2}\right)\alpha^2 \sum_{\ell=1}^K c_{\ell k} \bbE d^2(\bphi_{k, t+1}, \bphi_{\ell, t+1}) + \frac{\mu^2L}{2}\sigma^2
    \,,
\end{align}
where the second inequality uses the Cauchy-Schwarz inequality.
Now we need to upper bound $\bbE\|\nabla J_k(\bphi_{k, t+1})\|^2$. Let us consider
\begin{align}
\label{eq:bound_2}
    &\quad \ \bbE\|\nabla J_k(\bphi_{k, t+1})\|^2 \nonumber \\
    &= \bbE\|\nabla J_k(\bphi_{k, t+1})  - \Gamma_{\bw_{k, t}}^{\bphi_{k, t+1}} \nabla J_k(\bw_{k, t}) + \Gamma_{\bw_{k, t}}^{\bphi_{k, t+1}} \nabla J_k(\bw_{k, t}) \|^2 \nonumber \\
    &\leq 2\bbE\|\nabla J_k(\bphi_{k, t+1})  - \Gamma_{\bw_{k, t}}^{\bphi_{k, t+1}} \nabla J_k(\bw_{k, t})\|^2 \nonumber \\ & \quad + 2\bbE\| \nabla J_k(\bw_{k, t}) \|^2 \nonumber \\
    &\leq 2\mu^2L^2 \bbE\|\widehat{\nabla J}_k(\bw_{k, t})\|^2 + 2\bbE\| \nabla J_k(\bw_{k, t}) \|^2 \nonumber \\
    &\leq 2(\mu^2 L^2+1) \bbE\|{\nabla J}_k(\bw_{k, t})\|^2 + 2\mu^2 L^2 \sigma^2\,,
\end{align}
where the first inequality uses the fact that the parallel transport is isometric and $\frac{1}{2}(a+b)^2 \leq a^2 + b^2$, the second inequality uses~\eqref{eq:smooth_grad_cost}, and the last inequality follows from the orthogonality decomposition used in~\eqref{eq:smooth_cost_expectation} together with the bound~\eqref{eq:gradient_variance_square}.
Plugging the upper bound of $\frac{1}{2}\bbE\|\nabla J_k(\bphi_{k, t+1})\|^2$ provided in~\eqref{eq:bound_2} into~\eqref{eq:smooth_both} and reordering, we have
\begin{align}
\label{eq:smooth_both_2}
      & \quad \ \bbE J_k(\bw_{k, t+1}) \nonumber \\ & \leq \bbE J_k(\bw_{k, t}) - \left(\epsilon - \xi (\mu^2 L^2+1)\right) \bbE \| {\nabla J_k}(\bw_{k, t})\|^2 \nonumber \\ & \quad + \left(\frac{1}{2\xi}  + \frac{L}{2}\right)\alpha^2 \sum_{\ell=1}^K c_{\ell k} \bbE d^2(\bphi_{k, t+1}, \bphi_{\ell, t+1})
      \nonumber \\ & \quad + \left(\frac{\mu^2L}{2} + \xi\mu^2 L^2\right)\sigma^2 \nonumber \\
      &  = \bbE J_k(\bw_{k, t}) - \frac{\epsilon}{2} \bbE \| {\nabla J_k}(\bw_{k, t})\|^2  \nonumber \\ & \quad + \left(\frac{\mu^2L^2+1}{\epsilon}  + \frac{L}{2}\right)\alpha^2 \sum_{\ell=1}^K c_{\ell k} \bbE d^2(\bphi_{k, t+1}, \bphi_{\ell, t+1}) \nonumber \\ & \quad + \left(\frac{\mu^2L}{2} + \frac{\epsilon\mu^2 L^2}{2(\mu^2L^2+1)}\right)\sigma^2 \,,
\end{align}
where in the equality we select $\xi = \frac{\epsilon}{2(\mu^2L^2+1)}$ for simplicity. Since $\mu \in (0, \frac{1}{L}]$, we have $L \leq \mu^{-1}$ and $\frac{\mu}{2} \leq \epsilon \leq \mu$, so that $\frac{\epsilon\mu^2 L^2}{2(\mu^2L^2+1)} \leq \frac{\mu^2L}{2}$, and thus we can further simplify~\eqref{eq:smooth_both_2} as
\begin{align}
\label{eq:smooth_both_3}
      \bbE J_k(\bw_{k, t+1}) &\leq  \bbE J_k(\bw_{k, t}) - \frac{\mu}{4} \bbE \| {\nabla J_k}(\bw_{k, t})\|^2 + \mu^2L\sigma^2 \nonumber \\ & \quad + \frac{9\alpha^2}{2\mu} \sum_{\ell=1}^K c_{\ell k} \bbE d^2(\bphi_{k, t+1}, \bphi_{\ell, t+1})
      \,.
\end{align}
Taking the average of~\eqref {eq:smooth_both_3} over $k$, and considering the fact $\|{\nabla J}(\bw_t)\|^2= \frac{1}{K^2}\sum_{k=1}^K \| {\nabla J_k}(\bw_{k, t})\|^2$, we obtain the desired result.
\endgroup
\end{proof}

\section{Proof of lemma~\ref{lemma:taylor_expansion}}
\label{appx:lemma:taylor_expansion}
\begin{proof}
\begingroup
Let
$\bvarphi_{\ell,t+1} \triangleq \exp_{\bw_{\ell,t}}\!\big(-\mu \nabla J_\ell(\bw_{\ell,t})\big)$ for notation simplicity and define the perturbation term
\begin{align}
\bdelta_{\ell,t}
\triangleq \exp_{\bw_{k,t}}^{-1}(\bphi_{\ell,t+1})-\exp_{\bw_{k,t}}^{-1}(\bvarphi_{\ell,t+1})\,.
\end{align}
Then,
\begin{align}
\exp_{\bw_{k,t}}^{-1}(\bphi_{\ell,t+1})
= \exp_{\bw_{k,t}}^{-1}(\bvarphi_{\ell,t+1})+\bdelta_{\ell,t}\,.
\end{align}
Moreover, by applying Lemma~\ref{lemma:lipschitz_exponential_map} twice and using the definition of the gradient noise process, we obtain
\begin{align}
\|\bdelta_{\ell,t}\|
&= \big\|\exp_{\bw_{k,t}}^{-1}(\bphi_{\ell,t+1})-\exp_{\bw_{k,t}}^{-1}(\bvarphi_{\ell,t+1})\big\| \nonumber\\
&\leq (1+C_\kappa B^2)\, d(\bphi_{\ell,t+1},\bvarphi_{\ell,t+1}) \nonumber\\
&\leq (1+C_\kappa B^2)^2 \mu \big\|\widehat{\nabla J}_\ell(\bw_{\ell,t})-\nabla J_\ell(\bw_{\ell,t})\big\| \nonumber\\
&= (1+C_\kappa B^2)^2 \mu \|\bs_{\ell,t+1}(\bw_{\ell,t})\|\,.
\end{align}
\endgroup
Since the exponential map is a diffeomorphism close to $\bzero$, we can compute the Taylor series of
\[
{\exp_{\bw_{k, t}}^{-1}(\bphi_{\ell, t+1})
= \exp_{\bw_{k, t}}^{-1}\Big(\exp_{\bw_{\ell, t}}\big( - \mu {\nabla J}_\ell(\bw_{\ell, t})\big)\Big)}
\]
around $\bzero$.
Let us define a function $F(\bx) \triangleq \exp^{-1}_{\bw_{k,t}}(\exp_{\bw_{\ell,t}}(\bx))$
to simplify the presentation of the composite exponential maps. Using \emph{Taylor's theorem} (see Appendix A.6 of~\cite{absil2009manoptBook}), and considering the Taylor expansion at $\bzero$, evaluated at {$\bx=- \mu {{\nabla J}_\ell}(\bw_{\ell,t})$}, we have:
\begin{align}
\label{eq:tayloe_expansion_F}
    F\big( {- \mu {{\nabla J}_\ell}(\bw_{\ell,t})} \big) & = F(\bzero) + DF(\bzero)\big[\bx \big]  + R_{\ell, t}\,.
\end{align}
We now upper bound the three terms on the RHS of~\eqref{eq:tayloe_expansion_F}.
The first term is simply the value of $F$ at $\bzero$:
\begin{align}
    F(\bzero) = \exp_{\bw_{k,t}}^{-1}(\bw_{\ell,t})\,.
\end{align}
For the second term, we compute the differential of $F$ at $\bzero$ using the chain rule:
\begin{align}
    DF(\bzero) &= D\big(\exp^{-1}_{\bw_{k,t}}\circ \exp_{\bw_{\ell,t}}(\bzero) \big) \nonumber \\
    & = D \exp^{-1}_{\bw_{k,t}}(\exp_{\bw_{\ell,t}}(\bzero)) \circ D\exp_{\bw_{\ell,t}}(\bzero) \nonumber
    \\
    &= [D \exp_{\bw_{k,t}}( \exp^{-1}_{\bw_{k,t}}(\exp_{\bw_{\ell,t}}(\bzero)))]^{-1} \circ I_{T_{\bw_{\ell,t}}\calM} \nonumber \\
    &= [D \exp_{\bw_{k,t}}( \exp^{-1}_{\bw_{k,t}}(\bw_{\ell,t}))]^{-1}  \nonumber \\
    &= [\Lambda_{\bw_{k,t}}^{\bw_{\ell,t}}]^{-1} \,,
\end{align}
where the last equality follows the definition of the vector transport.
The third term can be written as:
\begin{align}
    \quad \ R_{\ell, t} \triangleq  \int_0^1 (1-s) D^2 F(s\bx) \big[\bx, \bx \big] ds\,.
\end{align}
Since $F$ is locally smooth, the operator norm of its Hessian tensor is bounded on a neighborhood of $\bzero$ by some constant $C_F>0$. Thus, the norm of $R_{\ell, t}$ can be bounded as:
\begin{align}
    {\|R_{\ell, t}\| \leq C_F\|\bx\|^2 = \mu^2 C_F {\|{{\nabla J}_\ell}(\bw_{\ell,t})\|^2}\,.}
\end{align}
In summary, we obtain the desired result.
\end{proof}

\section{Proof of lemma~\ref{lemma:consensus_bias_bound}}
\label{appx:lemma:consensus_bias_bound}
\begin{proof}
From Lemma~\ref{lemma:lipschitz_exponential_map}, we can upper bound the term $d^2(\bphi_{k, t+1}, \bphi_{\ell, t+1})$ as follows:
\begin{align}
    \label{eq:upper_bound_disagreement}
    & \quad \ d^2(\bphi_{k, t+1}, \bphi_{\ell, t+1}) \nonumber \\
    &\leq \left(1+C_{\kappa}B^2\right)^2
    \left\|\exp_{\bw_{k, t}}^{-1}(\bphi_{\ell, t+1}) - \exp_{\bw_{k, t}}^{-1}(\bphi_{k, t+1})\right\|^2 \nonumber \\
    &= \left(1+C_{\kappa}B^2\right)^2
    \left\|\exp_{\bw_{k, t}}^{-1}(\bphi_{\ell, t+1}) + \mu \widehat{\nabla J}_k(\bw_{k, t}) \right\|^2
    \nonumber \\
    &\leq {2\left(1+C_{\kappa}B^2\right)^2
    \left\|\exp_{\bw_{k, t}}^{-1}(\bphi_{\ell, t+1}) + \mu {\nabla J}_k(\bw_{k, t}) \right\|^2}
    \nonumber \\
    &\quad + {2\mu^2\left(1+C_{\kappa}B^2\right)^2 \left\| \bs_{k,t+1}(\bw_{k,t}) \right\|^2}\,.
\end{align}
{where first equality follows from the adaptation step in~\eqref{eq:diffusion} and the last inequality is due to $\frac{1}{2}(a+b)^2\leq a^2+b^2$}.
Combining the Taylor expansion~\eqref{eq:taylor_expansion} in Lemma~\ref{lemma:taylor_expansion} with~\eqref{eq:upper_bound_disagreement} we can obtain~\eqref{eq:upper_bound_disagreement_2}, shown at {the top of the next page},
where we use the fact that parallel transport is isometric.

To further upper bound the term $d^2(\bphi_{k, t+1}, \bphi_{\ell, t+1})$ in~\eqref{eq:upper_bound_disagreement_2}, it remains to show that the composite transport $[\Lambda_{\bw_{k,t}}^{\bw_{\ell,t}}]^{-1} \Gamma_{\bw_{k,t}}^{\bw_{\ell,t}}$ is locally close to the identity as in Theorem A.2.9 of~\cite{waldmann2012geometric} and~\cite{han2023riemannian}. This is also verified in Lemma 6 of~\cite{tripuraneni2018averaging} for general retractions. Consider the function $H(u) \triangleq [\Lambda_{x}^{\exp_x(u)}]^{-1} \Gamma_{x}^{\exp_x(u)}$, for any $v\in T_x\calM$, its evaluation is given by $H(u)[v] \in {\mathcal{L}}(T_x\calM)$, where ${\mathcal{L}}(T_x\calM)$ denotes the set of linear maps on $T_x\calM$. Let us apply Taylor's theorem for $H$ up to first order,  from \emph{Taylor's theorem} (see Appendix A.6 of~\cite{absil2009manoptBook}) we have
\begin{align}
    H(u)[v] = v + R_u\,,
\end{align}
where $\|R_u\| = C_R\|u\|^2$ with $C_R$ being a constant that is relevant to the Riemann curvature tensor~\cite{waldmann2012geometric, han2023riemannian} and depends on the smoothness of the exponential mapping, see the proof of Theorem 7 in~\cite{tripuraneni2018averaging}. Let $x = \bw_{k,t}$  and $u = \exp_{\bw_{k,t}}^{-1}(\bw_{\ell,t})$, then for {$v  = \Gamma_{\bw_{\ell,t}}^{\bw_{k,t}} {{\nabla J}_\ell}(\bw_{\ell,t})$}, we have
\begin{align}
    \label{eq:taylor_expansion_H}
    {[\Lambda_{\bw_{k,t}}^{\bw_{\ell,t}}]^{-1} \Gamma_{\bw_{k,t}}^{\bw_{\ell,t}}\Gamma_{\bw_{\ell,t}}^{\bw_{k,t}} {{\nabla J}_\ell}(\bw_{\ell,t}) = \Gamma_{\bw_{\ell,t}}^{\bw_{k,t}} {{\nabla J}_\ell}(\bw_{\ell,t})} + R_u\,,
\end{align}
and the norm of $R_u$ can be bounded as
\begin{align}
    \label{eq:residual_bound_H}
    \|R_u\| \leq C_R d^2(\bw_{k, t}, \bw_{\ell, t})\,.
\end{align}
\begin{figure*}[ht]
\noindent\hrulefill
\begin{align}
    \label{eq:upper_bound_disagreement_2}
    d^2(\bphi_{k, t+1}, \bphi_{\ell, t+1})
    & \leq  2\left(1+C_{\kappa}B^2\right)^2
    \left(\left\|\exp^{-1}_{\bw_{k, t}}(\bw_{\ell, t}) + \mu {\nabla J}_k(\bw_{k, t}) - \mu[\Lambda_{\bw_{k,t}}^{\bw_{\ell,t}}]^{-1} \Gamma_{\bw_{k,t}}^{\bw_{\ell,t}}\Gamma_{\bw_{\ell,t}}^{\bw_{k,t}} {{\nabla J}_\ell}(\bw_{\ell,t}) + R_{\ell, t} + {\bdelta_{\ell,t}} \right\|^2 \right) \nonumber\\
    & \quad + {2\mu^2\left(1+C_{\kappa}B^2\right)^2 \left\| \bs_{k,t+1}(\bw_{k,t}) \right\|^2}\,.
\end{align}
\noindent\hrulefill
\begin{align}
    \label{eq:upper_bound_disagreement_3}
    d^2(\bphi_{k, t+1}, \bphi_{\ell, t+1})
    & \leq {8}\left(1+C_{\kappa}B^2\right)^2
    \left(d^2(\bw_{k, t}, \bw_{\ell, t}) + {\mu^2\left\|{\nabla J}_k(\bw_{k, t}) - \Gamma_{\bw_{\ell,t}}^{\bw_{k,t}} {{\nabla J}_\ell}(\bw_{\ell,t})\right\|^2} + \mu^2\|R_u\|^2 + \|R_{\ell, t} + {\bdelta_{\ell,t}}\|^2 \right) \nonumber\\
    & \quad + {2\mu^2\left(1+C_{\kappa}B^2\right)^2 \left\| \bs_{k,t+1}(\bw_{k,t}) \right\|^2}\,.
\end{align}
\end{figure*}
From the Taylor expansion in~\eqref{eq:taylor_expansion_H} we can rewrite~\eqref{eq:upper_bound_disagreement_2} as in~\eqref{eq:upper_bound_disagreement_3}, shown at {the top of the next page},
where we use the fact $\frac{1}{4}(a+b+c+d)^2 \leq a^2 + b^2 + c^2 + d^2$.
We now upper bound the weighted summation and expectation of the four terms on the RHS of~\eqref{eq:upper_bound_disagreement_3}, using the fact that $C$ is symmetric and doubly stochastic.  \\
(i) For the first term on the RHS of~\eqref{eq:upper_bound_disagreement_3}, we have:
\begin{align}
    \label{eq:term_first}
    &\quad \ \bbE\sum_{k=1}^K\sum_{\ell=1}^K c_{\ell k}  d^2(\bw_{k, t}, \bw_{\ell, t}) \nonumber \\
    &\leq \sum_{k=1}^K\sum_{\ell=1}^K c_{\ell k}
    \Big(\bbE d^2(\bw_{k, t}, \bw_{m, t}) + \bbE d^2(\bw_{\ell, t}, \bw_{m, t})\Big) \nonumber \\
    &= 2\sum_{k=1}^K \bbE d^2(\bw_{k, t}, \bw_{m,t})
    = 2\bbE V_F(\bw_{t})\,.
\end{align}\\
(ii) For the second term on the RHS of~\eqref{eq:upper_bound_disagreement_3}, we have:
\begin{align}
    \label{eq:term_second}
    &\quad \  \bbE\sum_{k=1}^K\sum_{\ell=1}^K c_{\ell k} {\mu^2\left\|{\nabla J}_k(\bw_{k, t}) - \Gamma_{\bw_{\ell,t}}^{\bw_{k,t}} {{\nabla J}_\ell}(\bw_{\ell,t})\right\|^2} \nonumber \\
    & {\leq \mu^2 L^2\bbE\sum_{k=1}^K\sum_{\ell=1}^K c_{\ell k}  d^2(\bw_{k, t}, \bw_{\ell, t})} \nonumber \\
    & {=2\mu^2 L^2 \bbE V_F(\bw_{t}) } \,,
\end{align}
{where the frist inequality holds from the Lipschitz gradient in Assumption~\ref{assumption:smooth} and the second inequality follows from the result in~\eqref{eq:term_first}.}\\
\noindent (iii) For the third term on the RHS of~\eqref{eq:upper_bound_disagreement_3}, we have:
\begin{align}
\label{eq:term_third}
    \bbE\sum_{k=1}^K\sum_{\ell=1}^K c_{\ell k} \mu^2\|R_u\|^2
    & \leq \sum_{k=1}^K\sum_{\ell=1}^K c_{\ell k} \mu^2C_R^2\bbE d^4(\bw_{k, t}, \bw_{\ell, t}) \nonumber \\
    & \leq \mu^2 C_R^2 B^2\sum_{k=1}^K\sum_{\ell=1}^K c_{\ell k}\bbE d^2(\bw_{k, t}, \bw_{\ell, t})\nonumber\\
    & \leq 2\mu^2 C_R^2 B^2\bbE V_F(\bw_{t})
    \,,
\end{align}
where the first inequality follows from the bound on $\|R_u\|$ in~\eqref{eq:residual_bound_H}, the second uses the fact $d^2(\bw_{k, t}, \bw_{\ell, t})\leq B^2$ and the last holds from~\eqref{eq:term_first}. \\
\begingroup
({iv}) For the fourth term on the RHS of~\eqref{eq:upper_bound_disagreement_3}, we have:
\begin{align}
\label{eq:term_fourth}
    &\quad \ \bbE \sum_{k=1}^K\sum_{\ell=1}^K c_{\ell k} (R_{\ell, t} {+ \bdelta_{\ell,t}})^2
    \nonumber \\
    & = 2\bbE \sum_{\ell=1}^K R^2_{\ell, t}  + 2\bbE \sum_{\ell=1}^K \bdelta^2_{\ell,t} \nonumber \\
    & \leq {2} \mu^4 C_F^2 \bbE\sum_{k=1}^K \|{\nabla J}_k(\bw_{k,t})\|^4
    \nonumber \\
    & \quad {+2 \sum_{\ell=1}^K (1+C_\kappa B^2)^4 \mu^2  \bbE\big\{ \bbE\{ \| \bs_{\ell,t+1}(\bw_{\ell,t})\|^2| \calF_t\}\big\}}  \nonumber \\
    & \leq {2} \mu^4 C_F^2 KG^4 {+ 2 (1+C_\kappa B^2)^4 \mu^2 K\sigma^2 } \,,
\end{align}
where the first inequality follows from the fact $\frac{1}{2}(a+b)^2 \leq a^2 + b^2$ and the second inequality holds by the facts~\eqref{eq:residual_bound} and~\eqref{eq:perturbation_bound} in Lemma~\ref{lemma:taylor_expansion}, and the last holds from Lemma~\ref{lemma:gradient_bound} and the bound on the gradient noise in Assumption~\ref{assumption:gradient_noise}.\\
\noindent (v) For the fifth term on the RHS of~\eqref{eq:upper_bound_disagreement_3}, we have:
\begin{align}
\label{eq:term_fifth}
    &\quad \  \bbE\sum_{k=1}^K\sum_{\ell=1}^K c_{\ell k}  \frac{\mu^2}{4} \left\| \bs_{k,t+1}(\bw_{k,t}) \right\|^2
    \nonumber \\
    & = \frac{\mu^2}{4} \sum_{k=1}^K \bbE\big\{ \bbE\{\|\bs_{k, t+1}\|^2|\calF_{t}\} \big\} \leq \frac{\mu^2 K \sigma^2}{4}
\end{align}
where the last inequality follows the bound on the gradient noise in Assumption~\ref{assumption:gradient_noise}.
\endgroup

Combining the above five results in~\eqref{eq:term_first}-\eqref{eq:term_fifth} with the weighted summation and expectation of~\eqref{eq:upper_bound_disagreement_3}, we obtain the desired result.
\end{proof}

\section{Proof of Theorem~\ref{theorem:convergence}}
\label{appx:theorem:convergence}
\begin{proof}
Consider Theorem~\ref{theorem:network_agreement}, we can bound $\bbE V_F(\bw_{t})$ after sufficient number of iterations $t_o$ and thus rewrite the result in Lemma~\ref{lemma:consensus_bias_bound} as
\begin{align}
    \label{eq:consensus_bias_bound_2}
     \bbE P(\bphi_{t+1}) \leq {\calO(\mu^2\alpha^{-2}(\zeta_2-\zeta_1\alpha)^{-2})}\,,\quad \forall \, t\geq t_o\,.
\end{align}
From Lemma~\ref{lemma:cost_relation}, we have:
\begin{align}
    \label{eq:cost_relation_2}
    \bbE \| {\nabla J}(\bw_{t})\|^2 & \leq \frac{4}{\mu K} \bbE \left[J(\bw_{t}) - J(\bw_{t+1})\right] \nonumber \\ & \quad  + \frac{18\alpha^2}{\mu^2 K^2}\bbE P(\bphi_{t+1}) {{}+ \frac{4\mu L\sigma^2}{K}}\,.
\end{align}
Summarize~\eqref{eq:cost_relation_2} from $t = t_o+1$ to $T$, we have:
\begin{align}
    \label{eq:cost_relation_3}
    & \quad \
    \frac{1}{T-t_o}\sum_{t=t_o+1}^T\bbE \| {\nabla J}(\bw_{t})\|^2 \nonumber \\  &\leq \frac{4}{\mu K(T-t_0)} \bbE \left[J(\bw_{t_o+1}) - J(\bw_{T+1})\right] \nonumber \\ & \quad  +  \frac{18\alpha^2}{\mu^2 K^2(T-t_0)}\sum_{t=t_o+1}^{T} \bbE P(\bphi_{t+1}) {{}+ \frac{4\mu L\sigma^2}{K}}\,.
\end{align}
Combining the results in~\eqref{eq:consensus_bias_bound_2} and~\eqref{eq:cost_relation_3}, we obtain the desired result.
\end{proof}

\section{Proof of lemma~\ref{lemma:cost_relation_pl}}
\label{appx:lemma:cost_relation_pl}
\begin{proof}
    Using the {local} Riemannian PL condition in the descent inequality~\eqref{eq:cost_relation} of Lemma~\ref{lemma:cost_relation}, {for $t\geq t_p+1$}, we have:
\begin{align}
    \bbE J(\bw_{t+1}) & \leq \bbE J(\bw_{t}) - \frac{\mu K}{4\tau} \bbE \left\{J(\bw_{t}) - J(\bw^*)\right\}\nonumber \\ & \quad  + \frac{9\alpha^2}{2\mu K}\bbE P(\bphi_{t+1}) {{}+ \mu^2L\sigma^2} \,.
\end{align}
The proof follows by subtracting $J(\bw^*)$ from both sides of the inequality above.
\end{proof}

\section{Proof of Theorem~\ref{theorem:convergence_pl}}
\label{appx:theorem:convergence_pl}
\begin{proof}
We recursively apply~\eqref{eq:cost_relation_pl} in Lemma~\ref{lemma:cost_relation_pl} for $t \geq {\underline t} + 1$ to obtain:
\begin{align}
    &\quad \ \bbE \left\{J(\bw_{t}) - J(\bw^*)\right\} \nonumber \\ &\leq \left(1 - \frac{\mu K}{4\tau}\right)^{t-{\underline t}}\bbE \left\{J(\bw_{{\underline t}}) - J(\bw^*)\right\} \nonumber \\ & \quad  + \sum_{s={\underline t}+1}^t \left(1 - \frac{\mu K}{4\tau}\right)^{s-{\underline t}} \left(\frac{9\alpha^2}{2\mu K} \bbE P(\bphi_{s+1}) {{}+ \mu^2L\sigma^2}\right) \nonumber \\
    &\leq \left(1 - \frac{\mu K}{4\tau}\right)^{t-{\underline t}}\bbE \left\{J(\bw_{{\underline t}}) - J(\bw^*)\right\} \nonumber \\ & \quad  + {\calO((\zeta_2-\zeta_1\alpha)^{-2}) + \calO(\mu)} \,,
\end{align}
where the second inequality holds from the fact that $\sum_{s={\underline t}+1}^t \left(1 - \frac{\mu K}{4\tau}\right)^{s-{\underline t}} \leq \sum_{s={\underline t}+1}^\infty \left(1 - \frac{\mu K}{4\tau}\right)^{s-{\underline t}} \leq \frac{4\tau}{\mu K}${,} the result in~\eqref{eq:consensus_bias_bound_2}.
\end{proof}

\section{Grassmann manifold}
\label{appx:grassmann_manifold}
The Grassmann manifold $\GR$, a set of $p$-dimensional linear subspaces of $\bbR^n$, can be regarded as a smooth quotient manifold of the Stiefel manifold $\ST = \{\bU\in\bbR^{n\times p}: \bU^T\bU = \bI_p\}$, i.e., $\GR = \ST / \OG = \{\pi(\bU): \bU \in \ST\}$ where $\OG = \{\bO \in \bbR^{p\times p}: \bO^T\bO = \bI_p\}$ is the orthogonal group and $\pi: \ST \to \GR$ is the map $\pi(\bU) = \{\bU \bO:\bO \in \OG\}$.
The geodesic distance between two subspaces $\pi(\bU_1)$ and $\pi(\bU_2)$ of $\GR$, spanned by orthonormal matrices $\bU_{\!1}$ and $\bU_{\!2}$, is defined as follows~\cite{edelman1998geometry}:
\begin{equation}
    \label{eq:grassmann_distance}
	d_{\GR}(\bU_{\!1}, \bU_{\!2}) = \|\cos^{-1}(\btheta)\|_2 \,,
\end{equation}
where $\btheta\in\amsmathbb{R}^{p}$ contains the singular values of $\bU_{\!1}^T\bU_{\!2}$, namely, it
is related to its singular value decomposition (SVD) as $\bU_{\!1}^T\bU_{\!2}=\bV_{\!1}^T\diag(\btheta)\bV_{\!2}$.
Define $\bar f: \ST \to \bbR$, we have $f(\pi(\bU)) = \bar f (\bU)$ for all $\pi(\bU)\in \GR$. The Riemannian gradient $\nabla f$ at $\pi(\bU)\in \GR$ is given by
\begin{equation}
    \label{eq:grassmann_gradient}
	\nabla f(\pi(\bU)) = \nabla \bar f(\bU) = \bP_{\bU}^{\GR}(\bG)\,,
\end{equation}
with $\bP_{\bU}^{\GR}(\bG) = (\bI - \bU\bU^T)\bG$,
where $\bG \in \bbR^{n\times p}$ is the Euclidean gradient of $\bar f$ at $\bU$.
Let $\bxi \in T_{\pi(\bU)}\GR$ and $\bX\bSigma\bY=\bxi$ be the thin SVD of $\bU + \bxi\in\amsmathbb{R}^{n\times p}$.
Then, the exponential mapping is represented as~\cite{boumal2023introduction}
\begin{equation}
    \label{eq:grassmann_exp}
	\exp_{\pi(\bU)}(\bxi) = \bU\bY\cos(\bSigma) + \bX\sin(\bSigma) \,.
\end{equation}

\begingroup
\section{Fixed-rank manifold}
\label{appx:fixrank_manifold}
The fixed-rank manifold of rank-$r$ matrices in $\bbR^{m\times n}$ is defined by $\calM_r \triangleq \{\bA\in\bbR^{m\times n}: \mathrm{rank}(\bA)=r\}$.
Any point $\bA\in\calM_r$ admits a compact SVD representation $\bA = \bU\bS\bV^T$,
where $\bU\in\bbR^{m\times r}$ and $\bV\in\bbR^{n\times r}$ satisfy $\bU^T\bU=\bI_r$ and $\bV^T\bV=\bI_r$, while $\bS\in\bbR^{r\times r}$ is nonsingular.
Define $f: \calM_r \to \bbR$, and let $\bG\in\bbR^{m\times n}$ be the Euclidean gradient of $f$ at $\bA$. The Riemannian gradient $\nabla f$ at $\bA$ is represented by~\cite{vandereycken2013low}
\begin{equation}
    \label{eq:fixedrank_gradient}
    \nabla f(\bA)= \bU\bM\bV^T + \bU_p\bV^T + \bU\bV_p^T\,,
\end{equation}
where $\bM=\bU^T\bG^T\bV$, $\bU_p=\bG\bV-\bU\bM$ and $\bV_p=\bG^T\bU-\bV\bM^T$.
Let $\bxi \in T_{\bA}\calM_r$ with its unique representation $\bxi = \bU\dot{\bS}\bV^T + \bU_p \bV^T + \bU\bV_p^T$ where $\bU^T\bU_p=\boldsymbol{0}$ and $\bV^T\bV_p=\boldsymbol{0}$, then the orthographic retraction is given by~\cite{absil2012projection, absil2015low}
\begin{equation}
    \label{eq:fixedrank_retraction}
    R_{\bA}(\bxi) = \big(\bU(\bS+\dot{\bS}) + \bU_p\big)(\bS+\dot{\bS})^{-1}\big((\bS+\dot{\bS})\bV^T + \bV_p^T\big)\,.
\end{equation}
Given a nearby point $\bB\in\calM_r $, the inverse orthographic retraction at $\bA$ is given by~\cite{absil2015low}
\begin{equation}
    \label{eq:fixedrank_inv_retraction}
    R_{\bA}^{-1}(\bB) = \bB\bV\bV^T + \bU\bU^T\bB - \bU\bU^T\bB\bV\bV^T - \bA\,.
    \end{equation}
\endgroup

\bibliographystyle{IEEEtran}
\bibliography{manifoldOptRefs}

\end{document}